\documentclass[11pt]{article}

\usepackage[T1]{fontenc}
\usepackage[utf8]{inputenc}
\usepackage{times}
\usepackage{amsmath,amssymb,amsthm,mathtools}
\usepackage{enumitem}
\usepackage{microtype}
\usepackage[margin=1in]{geometry}
\usepackage{xurl}
\usepackage[colorlinks=true,allcolors=blue]{hyperref}

\newcommand{\F}{\mathbb F}
\newcommand{\Z}{\mathbb Z}

\newcommand{\R}{\mathbb R}

\newcommand{\ThreeSAT}{\mathrm{3SAT}}
\newcommand{\GapSVP}{\mathrm{GapSVP}}
\newcommand{\RS}{\mathrm{RS}}
\newcommand{\cC}{\mathcal C}
\newcommand{\cR}{\mathcal R}
\newcommand{\cL}{\mathcal L}

\newcommand{\wt}{\operatorname{wt}}

\newcommand{\eps}{\varepsilon}

\newtheorem{theorem}{Theorem}[section]
\newtheorem{lemma}[theorem]{Lemma}
\newtheorem{proposition}[theorem]{Proposition}
\newtheorem{corollary}[theorem]{Corollary}
\newtheorem{problem}{Problem}
\theoremstyle{remark}
\newtheorem{remark}[theorem]{Remark}

\setlist[itemize]{leftmargin=2em}
\setlist[enumerate]{leftmargin=2.3em}
\allowdisplaybreaks

\hypersetup{
  pdftitle={Polynomial-Factor Deterministic NP-Hardness for SVP in Every lp Norm with p > 2},
  pdfauthor={Isaac M Hair and Amit Sahai}
}

\title{Polynomial-Factor Deterministic NP-Hardness\\
for SVP in Every $\ell_p$ Norm with $p>2$}
\author{Isaac M Hair \thanks{\texttt{isaacmhair@gmail.com}} \\ UCSB, UCLA \and Amit Sahai \thanks{\texttt{sahai@cs.ucla.edu}} \\ UCLA}
\date{}

\begin{document}
\pagenumbering{roman}
\maketitle

\begin{abstract}
We prove that for every constant $2<p<\infty$, the shortest vector problem in the
$\ell_p$ norm is NP-hard to approximate within a polynomial factor under a
deterministic reduction. More precisely, for lattices of rank $M$, we obtain
a factor of $M^\varepsilon$ for every constant
$0<\varepsilon<\min\{(p-2)/(4p),1/8\}$. For $p=\infty$, the same holds for
every constant $0<\varepsilon<1/8$.

Our reduction adapts OpenAI's encoding of a satisfying assignment for the
closest vector problem~\cite{openai-ten-advances}. The encoding has one
block of coordinates for each evaluation of a polynomial, with one position
for each possible value. We allow the entries in a block to be arbitrary
field elements, which can cancel when we take weighted sums of powers of
these values. We show that on an unsatisfiable formula, a nonzero codeword
with small support must have many nonzero entries in each block it uses.
Applying a Hadamard matrix to each block, as in a work by Hair and Sahai~\cite{hair-sahai-2511}, gives the gap in vector length.

\emph{The main results of this paper were discovered by Codex / ChatGPT 5.6 Sol Ultra using harness components from the UCLA Moonshot Harness Project, in response to a general strategy suggested by the authors. See the end of the introduction for a discussion.}
\end{abstract}
\thispagestyle{empty}
\clearpage
\pagenumbering{arabic}

\section{Introduction}\label{sec:introduction}

Let $\mathbf B\in\Z^{M\times N}$ be a matrix with linearly independent rows.
The lattice generated by $\mathbf B$ is the set of their integer linear combinations,
\[
  \cL(\mathbf B)=\{\mathbf u\mathbf B:\mathbf u\in\Z^M\}.
\]
The lattice has rank $M$ and lies in $\R^N$. Lattices have applications in
algorithmic number theory~\cite{lenstra1982factoring}, integer
programming~\cite{lenstra1983integer,kannan1987minkowski,frank1987application,schrijver1998theory},
coding theory~\cite{forney1988coset,de1989some,zamir2014lattice}, and
cryptography~\cite{ajtai1996generating,nguyen2001two,micciancio2007worst,regev2009lattices,peikert2016decade}.
One reason lattices are of particular interest in cryptography is that the
security of some cryptosystems follows from the worst-case hardness of
approximating the length of a shortest nonzero lattice vector~\cite{regev2009lattices,peikert2009public}.

Given $\mathbf B$, the shortest vector problem in the $\ell_p$ norm
($\mathrm{SVP}_p$) asks for a shortest nonzero vector in $\cL(\mathbf B)$.
Write $\lambda_1^{(p)}(\cL(\mathbf B))$ for its length.
We study the corresponding decision problem.

\begin{problem}[$\gamma$-$\GapSVP_p$]
Let $\gamma=\gamma(M)\ge1$. Given a lattice basis $\mathbf B$ of rank $M$
and a positive rational threshold $r$, distinguish between the following two cases,
promised that one holds.
\begin{enumerate}[label=\textup{(\roman*)}]
  \item $\lambda_1^{(p)}(\cL(\mathbf B))\le r$.
  \item $\lambda_1^{(p)}(\cL(\mathbf B))>\gamma r$.
\end{enumerate}
\end{problem}

Both algorithms~\cite{kannan1987minkowski,ajtai2002sampling,blomer2009sampling,eisenbrand2020approximate}
and hardness reductions~\cite{micciancio2001shortest,haviv2007tensor,aggarwal2018gap,bennett2022improved,bhattiprolu2025inapproximability}
have been studied for a range of norms, as have reductions between
norms~\cite{regev2006lattice,aggarwal2021dimension}. For hardness, the
choice of norm made a difference from the very beginning.

\paragraph{A Brief History of SVP Hardness.}
Van Emde Boas~\cite{van1981another} proved that exact $\mathrm{SVP}_\infty$ is
NP-hard and asked whether the same holds for $\mathrm{SVP}_2$.
For $p=\infty$, deterministic reductions went on to give hardness of
approximation~\cite{arora1997hardness}, including a factor of
$M^{\Omega(1/\log\log M)}$ due to Dinur~\cite{dinur2002approximating}.
Mukhopadhyay~\cite{mukhopadhyay2022projection} obtained a polynomial factor
in this norm under the Projection Games Conjecture~\cite{moshkovitz2015projection}.\footnote{In a projection game, each constraint maps a label of one variable to a required label of another. The conjecture allows the fraction of constraints simultaneously satisfiable in a NO instance to be polynomially small.}

For finite $p$, the first hardness results used randomness. Ajtai~\cite{ajtai1998shortest}
showed that a polynomial-time algorithm for exact Euclidean SVP would imply
$\mathrm{NP}\subseteq\mathrm{RP}$.
Cai and Nerurkar~\cite{cai1998approximating} obtained a gap of one plus an
inverse polynomial, and Micciancio~\cite{micciancio2001shortest} obtained a
constant factor. Khot~\cite{khot2003hardness,khot2005hardness} ultimately showed
that, for every fixed $p>1$, a polynomial-time algorithm for any constant
approximation factor would imply $\mathrm{NP}\subseteq\mathrm{RP}$. Assuming NP is not contained in
randomized quasipolynomial time,
Khot and then Haviv and Regev obtained larger gaps, reaching
$2^{(\log M)^{1-\eta}}$ for every fixed $1\le p<\infty$ and
$\eta>0$~\cite{khot2005hardness,haviv2007tensor}.

The usual route started from the closest vector problem (CVP), which asks for a
lattice point closest to a given target. Deterministic reductions for CVP
already gave almost-polynomial gaps~\cite{arora1997hardness,dinur2003cvp}.
To pass from CVP to SVP, the reductions used a lattice with many points
near a target, even though its shortest nonzero vector is long.
Constructing such a \emph{locally dense lattice} deterministically was a
long-standing obstacle~\cite{micciancio2012inapproximability,micciancio2014locally,bennett2022hardness,bennett2023complexity}.

In earlier work~\cite{hair-sahai-2511}, we reduced a constraint satisfaction
problem directly to $\GapSVP_p$ and obtained deterministic NP-hardness for
a factor of $2^{(\log M)^{1-\eta}}$ for all constants $p>2$ and $\eta>0$.
A polynomial factor followed under the Sliding Scale Conjecture~\cite{bellare1993sliding},
which posits polynomial-size, constant-query PCPs with polynomially small
soundness error over a polynomial-size alphabet. Independently, Hecht and Safra~\cite{hecht-safra-2025} obtained
deterministic hardness results for $p>2$ under stronger complexity assumptions.
We subsequently obtained nearly polynomial factors for all finite norms
under a deterministic subexponential-time assumption~\cite{hair-sahai-2604}.
Wan~\cite{wan-2603} gave a deterministic construction of locally dense
lattices based on Reed--Solomon codes, building on Bennett and
Peikert~\cite{bennett2022hardness}. This gave deterministic NP-hardness
for every fixed finite $p\ge1$ and every $1\le\gamma<2^{1/p}$,
including exact Euclidean SVP. Wan later extended the Euclidean result
to every constant factor~\cite{wan-2608}.

\paragraph{Our Result.}
We obtain a polynomial factor for every $p>2$ by a deterministic
polynomial-time reduction from $\ThreeSAT$, without assuming the Sliding Scale Conjecture.

\begin{theorem}[Main Theorem]\label{thm:main-introduction}
For every constant $2<p<\infty$ and every constant
\[
  0<\varepsilon<\min\left\{\frac{p-2}{4p},\frac18\right\},
\]
$M^\varepsilon$-$\GapSVP_p$ on lattices of rank $M$ is NP-hard under a
deterministic polynomial-time reduction. For $p=\infty$, the same holds for
every constant $0<\varepsilon<1/8$.
\end{theorem}

\paragraph{From Closest Vectors to Shortest Vectors.}
Our starting point is OpenAI's reduction from $\ThreeSAT$ to Euclidean CVP
with a gap of $M^{1/400}$~\cite[Chapter~7]{openai-ten-advances}.
Liu~\cite{liu2026cvp} and Song~\cite{song2026closest} subsequently improved
the gap for CVP. In OpenAI's construction, a polynomial encodes an assignment.
The global table records its value at each evaluation point using a block
of bits, one for each possible value. The CVP target specifies the right-hand
side of an affine system on these bits. In particular, every global block
must have odd parity. The soundness proof can therefore count on finding
a nonzero entry at every evaluation point.

For SVP, we use a homogeneous system, so even a block in the
global table can vanish. There is also a problem with the modulus. The integer
lattice obtained by lifting a binary linear code always contains twice each
standard basis vector. We use a large prime modulus $q$ instead. This makes
every nonzero multiple of $q$ in $\Z^M$ long, but the entries in a block can
now be arbitrary elements of $\F_q$. Even when a block has nonzero entries, their sum in $\F_q$ can be zero.

We keep OpenAI's tests on the moments of each block, that is, on sums of
powers of the possible values weighted by the corresponding entries. We
prove that on an unsatisfiable formula, a nonzero codeword either has large
total support or has many nonzero entries in every block it uses. This is
where our earlier SVP reduction enters~\cite{hair-sahai-2511}. A block
containing a single $1$ maps to one Hadamard row. A block containing many
nonzero integers has large Euclidean norm, and orthogonality lets us use
that norm to lower-bound the length of its image. For $p>2$, the difference
between these two cases gives a polynomial gap.

\section*{Acknowledgements and AI Use Methodology}
\begin{flushleft}
\renewcommand{\arraystretch}{1.3}
\begin{tabular}{@{}l@{\qquad}c@{}}
\textbf{Activity} & \textbf{AI used} \\
Asking the research question & [No] \\
Coming up with the approach & [Yes] \\
Proof development & [Yes] \\
Writing and exposition & [Yes] \\
Checking for bugs & [Yes] \\
Other supporting tasks & [Yes] \\
\end{tabular}
\end{flushleft}

The human authors suspected that a combination of the polynomial moment ideas from the OpenAI CVP result~\cite{openai-ten-advances} and the Hadamard gadget ideas from a result by Hair and Sahai on SVP~\cite{hair-sahai-2511} could give polynomial factor NP hardness of approximation for SVP in all $\ell_p$ norms with $p > 2$. We investigated this idea with Codex / ChatGPT 5.6 Sol Ultra using harness components from the UCLA Moonshot Harness Project~\cite{MoonshotHarness}. After a few hours, Codex found a more complicated version of the proof presented in this paper. (One of Codex's important contributions was to generalize the code construction used for OpenAI's CVP hardness proof.) The human authors then carefully verified and refined all proof components, adding motivation and exposition to communicate the arguments clearly. The paper was written with the assistance of ChatGPT 5.6 Sol Ultra and ChatGPT 6 Astra Ultra.

Of course, the human authors take full responsibility for this paper and its contents. If any reader is aware of references that should be added, please contact us and we will update the manuscript.

This research was supported in part by a Laude Moonshot seed award, a Simons Investigator Award, a DARPA expMath award, NTT Research, NSF grant 2333935, BSF grant 2022370, a Xerox Faculty Research Award, a Google Faculty Research Award, an Okawa Foundation Research Grant, and the Symantec Chair of Computer Science.

\section{Preliminaries}\label{sec:preliminaries}
For a positive integer $k$, let $[k]=\{1,\ldots,k\}$.  All vectors are row
vectors unless stated otherwise.  For a vector $\mathbf u$ over any field, its
support and Hamming weight are
\[
  \operatorname{supp}(\mathbf u)=\{j:\mathbf u_j\ne0\},
  \qquad
  \wt(\mathbf u)=|\operatorname{supp}(\mathbf u)|.
\]

\begin{lemma}[Norm Comparison]\label{lem:norm-comparison}
Let $1\le r\le t\le\infty$, and let $\mathbf y\in\mathbb R^k$.  Then
\[
  \|\mathbf y\|_t\le \|\mathbf y\|_r
  \le k^{1/r-1/t}\|\mathbf y\|_t,
\]
where $1/\infty=0$.
\end{lemma}

\begin{proof}
The case $r=t$ is immediate. For $r<t$, the first inequality is monotonicity of the $\ell_p$ norms. For finite $t$,
the second follows by applying H\"older's inequality~\cite{holder1889ueber}
to $\sum_i |\mathbf y_i|^r$.  For $t=\infty$, it follows from
$\sum_i|\mathbf y_i|^r\le k\|\mathbf y\|_\infty^r$.
\end{proof}

\section{Technical Overview}\label{sec:overview}

We describe the reduction from $\ThreeSAT$~\cite{karp1972reducibility}
for a fixed finite rational $p>2$. We first build a linear code whose
coordinates are arranged in rows. A satisfying assignment gives a codeword
with a single $1$ in each row it uses. For an unsatisfiable formula we must rule out every short nonzero lattice
vector, even when entries in a row cancel in the moment tests.

\subsection{Encoding a Satisfying Assignment}

Let $s$ be the bit length of the input $3$CNF formula.  Delete tautological
clauses, repeated literals within a clause, and unused variables.  Write
$\varphi$ for the resulting formula, $x_1,\ldots,x_n$ for its variables,
and $C_1,\ldots,C_m$ for its clause occurrences.  We consider the case in
which $m\ge1$ and every clause is nonempty.  For each clause $C$, let $I_C$
be the set of indices of its variables and let $\mathrm{Sat}(C)$ be the
Boolean assignments to those variables that satisfy $C$.

We encode an assignment $\sigma\in\{0,1\}^n$ by a polynomial over a
prime field $\F_q$ with $q>n$. Associate a distinct field element
$a_i=i-1$ to each variable $x_i$.  We call these $n$ points the
\emph{anchors}.  There is a unique polynomial $f_\sigma\in\F_q[X]$ of
degree at most $n-1$ with $f_\sigma(a_i)=\sigma_i$ for every $i$.
We store its values at the remaining points
$P=\F_q\setminus\{a_1,\ldots,a_n\}$.  Our global table has a row for each
$\xi\in P$ and a column for each $w\in\F_q$.  In row $\xi$ we put a $1$
in column $f_\sigma(\xi)$ and put zeros everywhere else.

We also need to record how $\sigma$ satisfies each clause.  For each clause $C$ we make a separate copy of the table for every
$\beta\in\mathrm{Sat}(C)$.  If
$\sigma$ satisfies $\varphi$, the copy for $\beta=\sigma|_{I_C}$ agrees
with the global table, and the other copies for $C$ are zero.  Thus our
tables have labels
\[
  \Theta=\{\star\}\cup
    \{(C,\beta): C\text{ is a clause occurrence},\
                    \beta\in\mathrm{Sat}(C)\},
\]
where $\star$ stands for the global table.  We call these labels
\emph{types}, and call $(C,\beta)$ a subtype of $C$.  For a type
$\tau\in\Theta$, write $x_{\tau,\xi,w}$ for the entry in row $\xi$ and
column $w$.  The $q$ entries in this row form the \emph{evaluation fiber}
at $(\tau,\xi)$.  We call a fiber nonempty if it has a nonzero entry.

\paragraph{Linear Constraints on the Tables.}
We impose linear constraints on the table entries. In a codeword, the entry
$x_{\tau,\xi,w}\in\F_q$ is the weight assigned to $w$ in the fiber at
$(\tau,\xi)$. There are three families of constraints.

First, for each clause $C$, its tables must sum to the global table.
In coordinates, this says
\[
  x_{\star,\xi,w}=\sum_{\beta\in\mathrm{Sat}(C)}x_{(C,\beta),\xi,w}
      \qquad(\xi\in P,\ w\in\F_q).
\]
A satisfying assignment obeys this constraint because exactly one table
for $C$ copies the global table.

Second, the column indices selected by a satisfying assignment are values
of a single low-degree polynomial.  Power sums have also been used in
hardness reductions for Reed--Solomon decoding and for SVP~\cite{gandikota2018moments,bennett2022hardness}.  Here the
$j$th \emph{moment} of a fiber is the weighted sum of the $j$th powers
of its column indices,
\[
  \mu_{\tau,j}(\xi)=\sum_{w\in\F_q}x_{\tau,\xi,w}w^j,
\]
where $w^0=1$ even for $w=0$.  In a table that copies the encoding of
$\sigma$, this is just $f_\sigma(\xi)^j$.  Put $d=n$.  For every $j$
from $0$ up to a cutoff $T$ specified below, we require
$\mu_{\tau,j}(\xi)$ to agree on $P$ with a polynomial
$\mu_{\tau,j}(X)$ of degree at most $dj$.  This is a linear condition
on the entries.  In coding terminology, the vector of evaluations must
belong to the Reed--Solomon code $\RS(P,dj)$, the evaluations on $P$ of
polynomials of degree at most $dj$~\cite{reed-solomon}.

Third, a table for $(C,\beta)$ should describe an assignment with value
$\beta_i$ at $a_i$ for each $i\in I_C$.  If $f_\sigma(a_i)=\beta_i$,
then $X-a_i$ divides $f_\sigma(X)-\beta_i$.  Thus, in the active table,
replacing $w$ by $(w-\beta_i)/(\xi-a_i)$ still gives the value of a
polynomial, now of degree at most $d-1$.  This gives the
\emph{shifted moments}.  For $\tau=(C,\beta)$ and $i\in I_C$, require
\[
  \eta_{\tau,i,j}(\xi)=
  \sum_{w\in\F_q}x_{\tau,\xi,w}
       \left(\frac{w-\beta_i}{\xi-a_i}\right)^j
  \qquad(0\le j\le T)
\]
to agree on $P$ with a polynomial $\eta_{\tau,i,j}(X)$ of degree at most
$(d-1)j$.  The denominator is nonzero at every point of $P$.  An inactive
table satisfies both moment conditions because its weights are all zero.

\paragraph{The Number of Moments.}
Write $S=100+s+n+m$ and fix an integer $c\ge4$, which we will choose in
terms of the desired approximation factor. We want to recover the
information in a table from its fibers with at most $\kappa=S^c$ nonzero
entries.  For this we take moments through order $T=4S^{2c+1}$.  The
number of moments and the size of the field serve different purposes.
The bound $T\ge d\kappa^2$ will suffice to recover the Boolean values at the anchors.
Having many more than $dT$ evaluation points will let us turn
identities between evaluations into polynomial identities.  We take $Q=2c+3$ and
choose the first prime $q$ in $[S^Q,2S^Q)$.  In particular, $dT<|P|$,
so all the moment polynomials above are unique.

There are $g=|\Theta||P|$ fibers and $M=gq$ coordinates in all.  Our
three families of homogeneous linear constraints define a code
$\cC_\varphi\subseteq\F_q^M$, which we can construct by linear algebra.
The codeword $\mathbf x^\sigma$ given by a satisfying assignment has
exactly one $1$ in each of $R=(m+1)|P|$ nonempty fibers, one global fiber and
one fiber for each clause at every evaluation point.  The
fraction $R/g$ is between $1/8$ and $1$ since each clause has at most
eight satisfying assignments.  We will also use the bound $M<40S^D$,
where $D=2Q+1=4c+7$, which follows from $|\Theta|<8S$ and $|P|<q<2S^Q$.

\paragraph{The Support of a Short Codeword.}
A nonzero codeword can have several nonzero entries in a fiber, and their
weights can cancel even in the zeroth moment.  We prove that if the
formula is unsatisfiable and a nonzero codeword has Hamming weight at most
\[
  S_0=4M^\delta R,\qquad \delta=\frac{c-2}{D},
\]
then \emph{every} nonempty fiber of that codeword contains at least $T+2$
nonzero entries.  We explain the proof in Section~\ref{sec:techdecoding}.
We first use this statement to establish the lattice gap.

\subsection{From Hamming Weight to the Shortest Vector}
\label{sec:techgeometric}

We pass from the code to the lattice of all of its integer lifts,
\[
  \Lambda_\varphi=
  \{\mathbf z\in\Z^M:\mathbf z\bmod q\in\cC_\varphi\}.
\]
This is the usual Construction A~\cite{conway-sloane-1999}.  Any nonzero
$\mathbf z\in\Lambda_\varphi$ with zero residue has a coordinate of
magnitude at least $q$. For a vector with nonzero residue, we use either
its total support or the support of one fiber to lower-bound its norm.

As in our earlier work~\cite{hair-sahai-2511}, we apply a Hadamard
matrix separately to each fiber.  Let $L$ be the least power of two at
least $q$, and let $\mathbf U$ consist of the first $q$ rows of the $L\times L$
Sylvester Hadamard matrix~\cite{sylvester1867thoughts}.  Its entries
are $\pm1$, and $\mathbf U\mathbf U^{\mathsf T}=L\mathbf I_q$.
Our output lattice is
\[
  \Lambda'_\varphi=\{\mathbf z\mathbf E:\mathbf z\in\Lambda_\varphi\},
  \qquad \mathbf E=\mathbf I_g\otimes\mathbf U.
\]
The rows of $\mathbf U$ are orthogonal, so the map is injective.
Since $q\Z^M\subseteq\Lambda_\varphi$, the output lattice has rank $M$.
Its ambient dimension is $gL<2M$.

Consider first a satisfying assignment.  In the $0$--$1$ integer lift of its codeword, which we also denote by
$\mathbf x^\sigma$, every nonempty fiber is a standard basis vector.
Its image is therefore a single Hadamard row.  Across all $R$ nonempty
fibers, the resulting lattice vector has norm
\[
  \|\mathbf x^\sigma\mathbf E\|_p=(RL)^{1/p}=:r_0.
\]
Now consider an arbitrary integer fiber $\mathbf v\in\Z^q$.  Orthogonality
gives $\|\mathbf v\mathbf U\|_2=\sqrt L\|\mathbf v\|_2$.  By
Lemma~\ref{lem:norm-comparison}, this gives the first inequality below.  Applying it to every fiber and comparing norms on the vector
of the $g$ Euclidean fiber norms gives the second.
\[
  \|\mathbf v\mathbf U\|_p\ge L^{1/p}\|\mathbf v\|_2,
  \qquad
  \|\mathbf z\mathbf E\|_p\ge
    L^{1/p}g^{1/p-1/2}\|\mathbf z\|_2.
\]

Suppose now that $\varphi$ is unsatisfiable.  Write
$\overline{\mathbf z}=\mathbf z\bmod q$ for the codeword associated with
a nonzero $\mathbf z\in\Lambda_\varphi$.  There are three cases.
\begin{enumerate}[label=\textup{(\roman*)}]
\item If $\overline{\mathbf z}=0$, one fiber of $\mathbf z$ has
      Euclidean norm at least $q$.  The first bound gives
      $\|\mathbf z\mathbf E\|_p/r_0\ge q/R^{1/p}$.
\item If $\wt(\overline{\mathbf z})>S_0$, the integer vector
      $\mathbf z$ has Euclidean norm greater than
      $\sqrt{S_0}=2M^{\delta/2}R^{1/2}$.  The second bound gives
      $\|\mathbf z\mathbf E\|_p/r_0>
        2M^{\delta/2}(R/g)^{1/2-1/p}$.
\item If $0<\wt(\overline{\mathbf z})\le S_0$, some fiber contains at
      least $T+2$ nonzero residues.  Its integer lift has Euclidean norm
      at least $\sqrt{T+2}$, so the first bound gives
      $\|\mathbf z\mathbf E\|_p/r_0\ge\sqrt{T+2}/R^{1/p}$.
\end{enumerate}

Put $\vartheta=1/2-1/p$. Since
$R/g\ge1/8$, case~(ii) gives a constant times $M^{\delta/2}$.  In
case~(iii), we have $T=4S^{2c+1}$ and $R<Sq<2S^{2c+4}$.  Thus the ratio
in this case exceeds a constant times
$S^{c+1/2-(2c+4)/p}=S^{(2c+4)\vartheta-3/2}$.  Case~(i) gives a larger
power of $S$.  Using $M<40S^D$ and allowing room for the constant
factors and for rounding $r_0$ up to an integer, we obtain the exponent of $M$
\[
 \min\left\{\frac{(2c+4)\vartheta-7/4}{4c+7},
             \frac{2c-5}{4(4c+7)}\right\}.
\]
As $c$ grows, the two terms tend to $\vartheta/2$ and $1/8$,
respectively.  This gives every exponent claimed in
Theorem~\ref{thm:main-introduction} for finite rational $p$.
Section~\ref{sec:parameter-range} extends the result to real $p$ by norm
comparison.  For $p=\infty$ the same satisfying vector has norm exactly
$1$, and comparison with a sufficiently large finite norm gives the gap.

\subsection{Recovering an Assignment from a Codeword}
\label{sec:techdecoding}

We prove the claim used in case~(iii) by contraposition. Suppose a nonzero
codeword $\mathbf u$ has weight at most $S_0$ and has a fiber with between
$1$ and $T+1$ nonzero entries. We show how to recover a satisfying
assignment from $\mathbf u$.

\paragraph{Reconstructing the Support of Each Table.}
For a type $\tau$ and a point $\xi\in P$, let
$S_\tau(\xi)=\{w\in\F_q:\mathbf u_{\tau,\xi,w}\ne0\}$.
Call these column indices the \emph{nodes} of the fiber.  We first
retain only points in
\[
  P_\tau=\{\xi\in P: |S_\tau(\xi)|\le\kappa\}.
\]
Every point we discard accounts for more than $\kappa$ nonzero
coordinates.  Since $\wt(\mathbf u)\le S_0$, we discard at most
$S_0/\kappa<q/10$ points for each type.  The remaining points satisfy
$|P_\tau|>d(T+2\kappa(\kappa-1))$.  We will use this bound to establish
identities between polynomials by checking them at sufficiently many
points.

Consider one table, and let $h$ be the maximum number of nodes in a
retained fiber.  If $h=0$, all its moment polynomials vanish since
$dT<|P_\tau|$.  Suppose $h>0$ and consider the $h\times h$ matrix
whose $(i,j)$ entry is $\mu_{\tau,i+j}(X)$, for $0\le i,j<h$.  This
is the Hankel matrix of the moments.  At any retained point $\xi$ with
$h$ nodes, its determinant is
\[
   \det(\mathbf V_\xi)^2
   \prod_{w\in S_\tau(\xi)}\mathbf u_{\tau,\xi,w}\ne0,
\]
where $\mathbf V_\xi$ is the Vandermonde matrix with columns
$(1,w,\ldots,w^{h-1})^{\mathsf T}$ for the $h$ distinct nodes $w$.
The determinant uses all the weights multiplicatively.  In particular,
it detects the nodes even if their weights sum to zero.

Over the field $\F_q(X)$ of rational functions, the Hankel matrix is
invertible. We can therefore solve for $h$ coefficients that express each
moment of index $h,\ldots,2h-1$ as the same linear combination of its
$h$ predecessors.  The monic polynomial of this recurrence, which we denote by
$G_\tau(Z)$, specializes at each retained point with $h$ nodes to
$\prod_{w\in S_\tau(\xi)}(Z-w)$.  We call it the \emph{support polynomial}
of the table.  In particular, its roots are distinct, since its
discriminant (the product of the squared pairwise root differences) is
nonzero at such a point.  We take $G_\tau=1$ for $h=0$.

The recurrence holds through moment $T$ at every retained point
with $h$ nodes, because every node is a root of the specialized
support polynomial.  A retained point with fewer nodes makes the Hankel determinant vanish.
That determinant is nonzero and has degree at most $dh(h-1)$, so at most
$dh(h-1)$ retained points are lost. The bound on $|P_\tau|$ leaves enough
points with $h$ nodes to prove the recurrence by clearing denominators
and counting roots.  Writing $h_\tau=h$
for the degree of $G_\tau$ and $\alpha_{\tau,1},\ldots,\alpha_{\tau,h_\tau}$
for its roots in a field in which it splits, we obtain
\[
  G_\tau(Z)=\prod_{s=1}^{h_\tau}(Z-\alpha_{\tau,s}),
  \qquad
  \mu_{\tau,j}(X)=\sum_{s=1}^{h_\tau}\rho_{\tau,s}\alpha_{\tau,s}^{\,j}
  \quad(0\le j\le T).
\]
Indeed, the first $h$ moments determine the weights $\rho_{\tau,s}$ by a
Vandermonde system, and the recurrence determines the remaining moments.
Factoring the same Hankel determinant using these roots shows that
every $\rho_{\tau,s}$ is nonzero.  This is the weighted reconstruction
argument of Section~\ref{sec:weighted-reconstruction}.

\paragraph{Recovering a Boolean Value at an Anchor.}
Fix a subtype $\tau=(C,\beta)$ and a variable $i\in I_C$. We show that
evaluating $G_\tau$ at $X=a_i$ gives $(Z-\beta_i)^{h_\tau}$.
Write $t=X-a_i$. Since
the ordinary and shifted moments use the same weights, the binomial
theorem gives
\[
 t^j\eta_{\tau,i,j}
    =\sum_{r=0}^j\binom jr(-\beta_i)^{j-r}\mu_{\tau,r}.
\]
The two sides agree on $P_\tau$ and have degree at most $dj<|P_\tau|$,
so this is a polynomial identity.  Substituting the expression from the
previous paragraph shows that $\eta_{\tau,i,j}$ is the $j$th weighted
moment of the shifted nodes $y_s=(\alpha_{\tau,s}-\beta_i)/t$, with
the same weights $\rho_{\tau,s}$.

We would like to evaluate the support polynomial at $X=a_i$, but at this
point we have no way to evaluate the individual shifted nodes $y_s$ there.
It suffices to prove that
$H_i(Y)=\prod_s(Y-y_s)$ has coefficients in $\F_q[X]$.  Here is the
idea for $h=h_\tau>0$.  Form a column vector out of each $h$
consecutive shifted moments.  The linear map that advances these
vectors by one position has eigenvalues $y_1,\ldots,y_h$, by the same
Vandermonde factorization as above. We construct a basis in which
this map has polynomial entries.  Start with the $\F_q[X]$-span of
the first $h$ vectors and add one vector at a time.  Each time the span
grows strictly, the degree of the determinant of its basis drops by at
least one.  The degree starts at most $(d-1)h(h-1)$, and
$T\ge d\kappa^2$ supplies enough moments for two successive spans to
coincide.  The shift then preserves that span.  In a basis over
$\F_q[X]$, its matrix has polynomial entries, so its characteristic
polynomial $H_i$ has polynomial coefficients as well
(Section~\ref{sec:anchors}).

Since $G_\tau(Z)=t^hH_i((Z-\beta_i)/t)$ and $H_i$ is monic,
the leading term gives $(Z-\beta_i)^h$.  Every other term has a factor
of $t$.  This proves both that $G_\tau$ has polynomial coefficients and
that
\[
   G_\tau(Z)\bmod(X-a_i)=(Z-\beta_i)^h.
\]
When every retained fiber is empty, $h=0$ and the conclusion reads
$G_\tau=1$.
In particular, two subtypes that assign opposite bits to the same variable
have coprime support polynomials over $\F_q(X)$.  To see this, take the
resultant of the two monic polynomials.  It is a polynomial in their
coefficients that vanishes exactly when they have a common root.
Evaluating at $X=a_i$ gives the resultant of powers of $Z$ and $Z-1$,
which is nonzero.

\paragraph{Factoring the Global Support Polynomial.}
We now relate the global table to the subtypes for one clause $C$.
For each type $\tau$, form the rational function
\[
  \cR_\tau(Z)=\sum_{s=1}^{h_\tau}
                  \frac{\rho_{\tau,s}}{Z-\alpha_{\tau,s}}
             =\frac{A_\tau(Z)}{G_\tau(Z)}.
\]
Expanding $1/(Z-\alpha)$ as $Z^{-1}+\alpha Z^{-2}+\cdots$, we see
that the coefficient of $Z^{-j-1}$ is $\mu_{\tau,j}$ for $0\le j\le T$.
In particular, the numerator $A_\tau$ has coefficients in $\F_q(X)$,
as one sees by multiplying this expansion by $G_\tau$ and keeping the
terms of nonnegative degree in $Z$.  Also, $A_\tau$ and $G_\tau$ are
coprime.  Indeed, at each root $\alpha_{\tau,s}$ of $G_\tau$ the
numerator has value
\[
   \rho_{\tau,s}\prod_{r\ne s}(\alpha_{\tau,s}-\alpha_{\tau,r})\ne0.
\]
For $h_\tau=0$ the fraction is $0/1$.

The clause sums imply that the first $T+1$ coefficients of
$\cR_\star$ agree with those of
$\sum_{\beta\in\mathrm{Sat}(C)}\cR_{(C,\beta)}$.  The difference has a
common denominator of degree at most $9\kappa<T+1$, counting the
global table and up to eight subtypes.  A nonzero proper fraction with
denominator degree at most $9\kappa$ must have a nonzero term among
its first $9\kappa$ coefficients.  Thus the rational functions are
equal.  Two distinct subtypes for $C$ disagree on some variable, so
their denominators are pairwise coprime by the anchor argument.  Each
fraction is already in lowest terms, so no denominator factor cancels
in the sum.  Comparing monic denominators gives
\[
    G_\star=\prod_{\beta\in\mathrm{Sat}(C)}G_{(C,\beta)}
    \qquad\text{for every clause }C.
\]

\paragraph{Selecting a Satisfying Assignment.}
We now use the hypothesis that at least one fiber has between
$1$ and $T+1$ nonzero entries.  If $G_\star=1$, the factorization
above forces every $G_\tau$ to equal $1$, so all the moment polynomials
through order $T$ vanish.  Even if this particular fiber was discarded
during reconstruction, its moments are values of those same polynomials.
If the fiber has $k$ nodes, its first $k$ moments give an invertible
Vandermonde system with zero right-hand side.  This forces all $k$
weights to be zero, a contradiction.

We can therefore choose any irreducible factor of $G_\star$ over
$\F_q(X)$.  For every clause $C$, exactly one subtype denominator
contains this factor.  Choosing its label gives a satisfying assignment
to the variables of $C$.  If two chosen labels disagreed on a shared
variable, their support polynomials would be coprime, yet both contain
our chosen factor.  Hence the labels agree on shared variables and
together satisfy $\varphi$.  This proves the claim needed in case~(iii).
All of this reconstruction takes place in the soundness proof.  The
reduction itself only constructs the code and the lattice.

\paragraph{The Role of $p>2$.}
At $p=2$ the Hadamard map scales all Euclidean norms by $\sqrt L$.
Our satisfying assignment gives a vector of norm $\sqrt{RL}$.  In
case~(iii), the lower bound from one fiber is $\sqrt{(T+2)L}$.
Their ratio is $\sqrt{(T+2)/R}$, and in fact
\[
  \sqrt{\frac{T+2}{R}}<\frac{2\sqrt2}{S}<1,
\]
since $T+2<8S^{Q-2}$ and $R=(m+1)(q-n)>q\ge S^Q$.
For $p>2$ the same $R$ Hadamard rows together have norm $R^{1/p}L^{1/p}$,
whereas a single fiber with $T+2$ nonzero residues still contributes
at least $\sqrt{T+2}\,L^{1/p}$.  This difference accounts for the
factor $1/2-1/p$ in our exponent.

\section{A Homogeneous Weighted Moment Code}\label{sec:code}

We now give the code formally. We use the assignment tables and the ordinary
and shifted moments from the CVP construction~\cite[Chapter~7]{openai-ten-advances},
with homogeneous constraints over $\F_q$.

\paragraph{Parameters.}
Fix a rational $p>2$, and write $\vartheta=1/2-1/p$. We give the reduction
for a family of choices of the integer $c\ge4$. The parameter $\delta$
will determine which codewords we must decode, and
$\eps_c(p)$ will be the resulting exponent in the lattice gap.
Set
\begin{equation}\label{eq:parameter-family}
  Q=2c+3,
  \qquad D=2Q+1=4c+7,
  \qquad
  \chi=\min\left\{
    (2c+4)\vartheta-\frac74,
    \frac{2c-5}{4}
  \right\},
\end{equation}
and
\begin{equation}\label{eq:parameter-exponents}
  \delta=\frac{c-2}{D},\qquad
  \eps_c(p)=\frac\chi D.
\end{equation}
The two terms in the minimum come from the two ways a nonzero codeword can
give a long lattice vector. One estimate uses the number of nonzero coordinates in a fiber, and the
other uses the total support. Since $\vartheta>0$, we may and
do choose $c$ large enough that $\chi>0$.

\subsection{The Formula and Its Tables}\label{subsec:formula-tables}

Let $\varphi_{\rm in}$ be the input $3$CNF formula and $s$ its bit length.
Delete tautological clauses, repeated literals within a clause, and unused
variables to obtain $\varphi$. Write $x_1,\ldots,x_n$ for the remaining
variables and $C_1,\ldots,C_m$ for the clauses, counting different
occurrences of the same clause separately. For each clause $C$, write $I_C$
for its set of variable indices and $\mathrm{Sat}(C)\subseteq\{0,1\}^{I_C}$
for the assignments to these variables that satisfy $C$. We consider the
nontrivial case, in which $n,m\ge1$ and every clause is nonempty.

We use $d$ as a degree bound, $\kappa$ as a bound on the support of a fiber
to be reconstructed, and $T$ as the largest moment index. In terms of the
input size, these are
\begin{equation}\label{eq:base-parameters}
  S=100+s+n+m,\qquad
  d=n,\qquad
  \kappa=S^c,\qquad
  T=4S^{2c+1}.
\end{equation}
Thus $S\ge102$. Take $q$ to be the first prime in the range
\begin{equation}\label{eq:q-range}
  S^Q\le q<2S^Q.
\end{equation}
Such a prime exists by Bertrand's postulate~\cite{erdos1932beweis}, and we can
find it deterministically by testing the integers in order~\cite{aks}.
In $\F_q$, reserve one anchor $a_i=i-1$ for each variable $x_i$.
The anchors are distinct since $n<S<q$. We use all remaining field
elements as evaluation points, so
\[
  P=\F_q\setminus\{a_1,\ldots,a_n\}.
\]

The table types are
\begin{equation}\label{eq:types}
  \Theta=\{\star\}\cup
  \{(C_r,\beta):1\le r\le m,\ \beta\in \mathrm{Sat}(C_r)\}.
\end{equation}
There is one global table, indexed by $\star$. For each clause $C$ and
each $\beta\in\mathrm{Sat}(C)$ there is also a table indexed by the pair
$(C,\beta)$, which we call a subtype. Each table has one row for every
$\xi\in P$ and one column for every $w\in\F_q$. We write $x_{\tau,\xi,w}$
for the field element in the corresponding entry of table $\tau$, and call
its row $(x_{\tau,\xi,w})_{w\in\F_q}$ the evaluation fiber at $(\tau,\xi)$.

Write $g$ for the number of fibers and $M$ for the number of coordinates.
A satisfying assignment uses the global table and one subtype per clause.
Write $R$ for the number of fibers in these $m+1$ tables. Thus
\begin{equation}\label{eq:M-g-R}
  g=|\Theta||P|\quad\text{evaluation fibers},\qquad
  M=gq\quad\text{coordinates},\qquad R=(m+1)|P|.
\end{equation}
Since $1\le|\mathrm{Sat}(C)|\le8$,
\begin{equation}\label{eq:R-over-g}
  \frac18\le\frac Rg\le1,\qquad
  \frac{S^{2Q}}2<M<40S^D.
\end{equation}
Indeed, $m+1\le|\Theta|\le8(m+1)$ gives the first pair of bounds.  Also
$|P|=q-n>q/2$, so $M>q^2/2\ge S^{2Q}/2$.  In the other direction,
$|\Theta|<8S$, $|P|<q$, and $q<2S^Q$ give
$M<32S^{2Q+1}<40S^D$.

\subsection{The Linear Constraints}
\label{subsec:code-constraints}

For a finite set \(A\subseteq\F_q\) and an integer \(e\ge0\), define
\[
  \RS(A,e)
  =\{(f(\xi))_{\xi\in A}:
      f\in\F_q[X]\text{ and }\deg f\le e\}.
\]
This is the Reed--Solomon evaluation space of degree at most \(e\)
\cite{reed-solomon}.  If \(e<|A|\), the representing polynomial is unique,
because a nonzero polynomial of degree at most \(e\) cannot have more than
\(e\) roots.

First, require the tables belonging to each clause to sum to the global table.
That is, for every clause $C$, $\xi\in P$, and $w\in \F_q$, impose
\begin{equation}\label{eq:clause-decomposition}
  x_{\star,\xi,w}
  =\sum_{\beta\in \mathrm{Sat}(C)}x_{(C,\beta),\xi,w}.
\end{equation}
Second, we enforce a degree bound across the fibers of each table. For
$\tau\in\Theta$ and $0\le j\le T$, define the ordinary moment
\begin{equation}\label{eq:ordinary-moment}
  \mu_{\tau,j}(\xi)=\sum_{w\in \F_q}x_{\tau,\xi,w}w^j,
\end{equation}
where $w^0=1$ even at $w=0$, and require
\begin{equation}\label{eq:ordinary-rs}
  (\mu_{\tau,j}(\xi))_{\xi\in P}\in\RS(P,dj).
\end{equation}
Finally, consider a subtype $\tau=(C,\beta)$. If an assignment polynomial
$f$ satisfies $f(a_i)=\beta_i$ for $i\in I_C$, then $X-a_i$ divides
$f(X)-\beta_i$. To express this condition through moments, for each
$i\in I_C$ and $0\le j\le T$ define the shifted moment
\begin{equation}\label{eq:shifted-moment}
  \eta_{\tau,i,j}(\xi)
  =\sum_{w\in \F_q}x_{\tau,\xi,w}
     \left(\frac{w-\beta_i}{\xi-a_i}\right)^j
\end{equation}
and require
\begin{equation}\label{eq:shifted-rs}
  (\eta_{\tau,i,j}(\xi))_{\xi\in P}
  \in\RS(P,(d-1)j).
\end{equation}
Here $\xi\in P$, so $\xi-a_i\ne0$, and we use the same convention for
zeroth powers as above.

Let $\cC_\varphi\le\F_q^M$ be the common solution space of
\eqref{eq:clause-decomposition}, \eqref{eq:ordinary-rs}, and
\eqref{eq:shifted-rs}. All three conditions are homogeneous and linear
in the table entries. To construct the checks for $\RS(P,e)$, form the
matrix whose row at $\xi$ is $(1,\xi,\ldots,\xi^e)$. Its column space
consists of the evaluation vectors of polynomials of degree at most $e$. A basis for
its left nullspace therefore gives a parity-check matrix for the evaluation
code. In particular, we can construct the whole system by finite-field
linear algebra.

We record the inequalities on the degrees and the number of moments that
we will use in the proof.
\begin{equation}\label{eq:moment-budget}
  dT<|P|,\qquad T\ge d\kappa^2,\qquad T>9\kappa.
\end{equation}
Indeed, $d=n\le S$, $S\ge102$, and $q\ge S^{2c+3}$ give
$dT\le4S^{2c+2}\le4q/S<q-n=|P|$. Also
$d\kappa^2\le S^{2c+1}<T$ and $9\kappa<T$. The first inequality in
\eqref{eq:moment-budget} implies that the vectors in
\eqref{eq:ordinary-rs} and \eqref{eq:shifted-rs} determine unique
polynomials of their respective degree bounds. We write $\mu_{\tau,j}(X)$
and $\eta_{\tau,i,j}(X)$ for these polynomials, and $\mu_{\tau,j}(\xi)$
and $\eta_{\tau,i,j}(\xi)$ for their evaluations.

\subsection{The Codeword for a Satisfying Assignment}\label{subsec:honest-codeword}

Suppose $\sigma\in\{0,1\}^n$ satisfies $\varphi$.  Let
$f_\sigma\in \F_q[X]$ be the polynomial of degree at most $n-1$ satisfying
\begin{equation}\label{eq:assignment-interpolation}
  f_\sigma(a_i)=\sigma_i\qquad(1\le i\le n).
\end{equation}
For every $\xi\in P$, place a $1$ at $w=f_\sigma(\xi)$ in the global
evaluation fiber.  For each clause $C$, place a $1$ at the same $w$ in the unique
subtype $(C,\sigma|_{I_C})$, and put zero everywhere else.  Denote the
result by $\mathbf x^\sigma$.

\begin{lemma}[Code Completeness]\label{lem:code-completeness}
If $\sigma$ satisfies $\varphi$, then
\[
  \mathbf x^\sigma\in\cC_\varphi,\qquad \wt(\mathbf x^\sigma)=R.
\]
Exactly $R$ evaluation fibers contain one nonzero coordinate.
\end{lemma}

\begin{proof}
Equation~\eqref{eq:clause-decomposition} holds because for each clause
exactly one satisfying subtype copies the global entry.  In the global table and in each selected subtype the ordinary
moments are $f_\sigma(X)^j$, of degree at most $(n-1)j\le dj$.
For an active subtype and $i\in I_C$,
\[
  \frac{f_\sigma(X)-\sigma_i}{X-a_i}
\]
is a polynomial of degree at most $n-2$ when $n\ge2$, and is zero when
$n=1$.  Its $j$th power gives the required shifted moment. All ordinary and shifted
moments of the unselected subtypes vanish, so they too satisfy the degree
bounds.  At each point of
$P$, there is one nonempty global evaluation fiber and one nonempty subtype
evaluation fiber per clause, giving $(m+1)|P|=R$ nonzero coordinates.
\end{proof}

\section{Weighted Polynomial-Moment Reconstruction}
\label{sec:weighted-reconstruction}

At each evaluation point $\xi$ in a set
$P'\subseteq\F_q$, let $S(\xi)$ be the set of column indices with a
nonzero entry. We call these indices the \emph{nodes} of the fiber and
write $b_{\xi,w}\in\F_q^\times$ for the entry, or \emph{weight}, at node
$w$. The $j$th moment of the fiber is
\[
 \sum_{w\in S(\xi)} b_{\xi,w}w^j .
\]
Here we use $w^0=1$, also for $w=0$. Thus the zeroth moment is the sum
of the weights, and it may vanish even if $S(\xi)$ is nonempty.

For a single fiber, Prony's method recovers distinct nodes by finding
a linear recurrence for their weighted power sums~\cite{prony-1795}.
The same Vandermonde calculation appears in the sparse interpolation
algorithm of Ben-Or and Tiwari~\cite{ben-or-tiwari-1988}.
OpenAI's reconstruction across evaluation points treats the case in which all
weights equal one~\cite[Chapter~7, Lemma~10]{openai-ten-advances}.
Here both the supports and their nonzero weights can vary with the point. We will show that if these supports are small
at sufficiently many points, a single recurrence over the rational function
field $\F_q(X)$ determines the moment polynomials up to index $T$ from its
initial values. The roots of its polynomial may lie in a finite
extension of $\F_q(X)$. We need their existence only for the soundness
proof.

We recall how to recognize that a monic polynomial has distinct roots
after we extend the field. Let $F(Z)=Z^h+c_{h-1}Z^{h-1}+\cdots+c_0$
with $h\ge1$. In a splitting field, that is, a field generated over the
coefficient field by a full set of roots of $F$, list its roots
$\lambda_1,\ldots,\lambda_h$ with multiplicity. Its discriminant is
\[
 \operatorname{Disc}(F)=\prod_{1\le r<s\le h}(\lambda_s-\lambda_r)^2.
\]
It is nonzero precisely when the roots are distinct. The product is
symmetric in the roots, so the fundamental theorem of symmetric polynomials
and Vieta's formulas express it as a polynomial with integer coefficients
in $c_0,\ldots,c_{h-1}$~\cite{cox-little-oshea-2015}. For example, the
discriminant of $Z^2+c_1Z+c_0$ is $c_1^2-4c_0$. We will use the
coefficient formula in the following way. If the coefficients of $F$ are
rational functions of $X$ defined at $\xi$, then specializing the
discriminant at $X=\xi$ gives the discriminant of the specialized
polynomial.

\begin{lemma}[Weighted Polynomial-Moment Reconstruction]
\label{lem:weighted-reconstruction}
Let $d,\kappa\ge1$ and $T\ge0$ be integers, and let $P'\subseteq \F_q$.  At
every $\xi\in P'$, let $S(\xi)\subseteq \F_q$ have size at
most $\kappa$, and attach a nonzero weight $b_{\xi,w}\in \F_q^\times$ to each
$w\in S(\xi)$.  Suppose that, for $0\le j\le T$, polynomials
$\mu_j\in \F_q[X]$ satisfy
\begin{equation}\label{eq:reconstruction-1}
 \deg\mu_j\le dj,
 \qquad
 \mu_j(\xi)=\sum_{w\in S(\xi)}b_{\xi,w}w^j.
\end{equation}
Assume
\begin{equation}\label{eq:reconstruction-2}
 T\ge2\kappa-1,
 \qquad
 |P'|>d\bigl(T+2\kappa(\kappa-1)\bigr).
\end{equation}
In particular $P'$ is nonempty. Let
\[
 h=\max_{\xi\in P'}|S(\xi)|.
\]
Then there is a monic polynomial $G\in\F_q(X)[Z]$ of degree $h$ with the
following property.  In a splitting field $\Omega$ of $G$,
\[
 G(Z)=\prod_{s=1}^h(Z-\alpha_s)
\]
for pairwise distinct nodes $\alpha_1,\ldots,\alpha_h\in\Omega$, and there
are nonzero weights $\rho_1,\ldots,\rho_h\in\Omega$ such that
\begin{equation}\label{eq:reconstruction-3}
 \mu_j(X)=\sum_{s=1}^h\rho_s\alpha_s^j
 \qquad (0\le j\le T).
\end{equation}
When $h=0$, take $G=1$ and $\Omega=\F_q(X)$.  The conclusion means that every
$\mu_j$ is zero.
\end{lemma}

The polynomial $G$ is the \emph{support polynomial}. Its roots and weights give a
single representation of the moment polynomials in \eqref{eq:reconstruction-3}. The
roots and weights belong to the extension $\Omega$, while the coefficients
of $G$ belong to $\F_q(X)$.

\begin{proof}
If $h=0$, all fibers are empty. By \eqref{eq:reconstruction-1}, each $\mu_j$ vanishes
on $P'$. Since $\deg\mu_j\le dT<|P'|$ by \eqref{eq:reconstruction-2}, we have
$\mu_j=0$ for all $j$. This gives the claim with $G=1$ and
$\Omega=\F_q(X)$.

Now suppose $h\ge1$. We first count points with exactly $h$ nodes. For
this, consider the $h\times h$ matrix whose $(i,j)$ entry is the moment
of index $i+j$. Its determinant is the Hankel determinant
\[
 \Delta_h(X)=\det\bigl(\mu_{i+j}(X)\bigr)_{0\le i,j<h}.
\]
The largest moment index here is $2h-2\le2\kappa-2\le T$.  In a determinant
term indexed by a permutation $\pi$ of $\{0,\ldots,h-1\}$, the total degree
is at most
\[
 d\sum_{i=0}^{h-1}(i+\pi(i))=dh(h-1),
\]
and therefore
\begin{equation}\label{eq:reconstruction-5}
 \deg\Delta_h\le dh(h-1).
\end{equation}

Fix a point $\xi$ with $|S(\xi)|=h$, list its nodes as
$w_1,\ldots,w_h$, and put
\[
 (\mathbf V_\xi)_{i,s}=w_s^i,
 \qquad
 \mathbf W_\xi=\operatorname{diag}(b_{\xi,w_1},\ldots,b_{\xi,w_h}).
\]
The $s$th column of the Vandermonde matrix $\mathbf V_\xi$ lists the powers
$1,w_s,\ldots,w_s^{h-1}$.  Its determinant is
\(\prod_{r<s}(w_s-w_r)\), which is nonzero because the nodes are distinct.
By \eqref{eq:reconstruction-1}, the Hankel matrix at $\xi$ factors as
\[
 \bigl(\mu_{i+j}(\xi)\bigr)_{i,j<h}
 =\mathbf V_\xi\mathbf W_\xi\mathbf V_\xi^{\mathsf T}.
\]
Taking determinants and using the fact that all weights are nonzero gives
\[
 \Delta_h(\xi)=\det(\mathbf V_\xi)^2
   \prod_{s=1}^hb_{\xi,w_s}\ne0.
\]
At a point with fewer than $h$ nodes, the same Hankel matrix factors through
an $h\times |S(\xi)|$ matrix and has rank below $h$.  Thus
\(\Delta_h(\xi)=0\)
there.  In particular, $\Delta_h$ is a nonzero polynomial, and the number of
maximum-support points is at least
\begin{equation}\label{eq:reconstruction-7}
 |P'|-dh(h-1)
 >d\bigl(T+2\kappa(\kappa-1)-h(h-1)\bigr)
 \ge d\bigl(T+h(h-1)\bigr).
\end{equation}
The last inequality uses $h(h-1)\le\kappa(\kappa-1)$. We will use the
strict inequality in \eqref{eq:reconstruction-7} to prove a polynomial identity by
counting its roots.

At a point with $h$ nodes, their monic degree-$h$ polynomial gives a
recurrence for the local moments. We look for its coefficients over
$\F_q(X)$ by requiring the recurrence to hold for the first $h$ shifts.
That is, solve for $\gamma_0,\ldots,\gamma_{h-1}$ in
\begin{equation}\label{eq:reconstruction-8}
 \sum_{\ell=0}^{h-1}\mu_{i+\ell}\gamma_\ell=-\mu_{i+h}
 \qquad (0\le i<h).
\end{equation}
All moments used are available because $2h-1\le2\kappa-1\le T$, and the
coefficient determinant is the nonzero element $\Delta_h$.  Cramer's rule
gives
\begin{equation}\label{eq:reconstruction-9}
 \gamma_\ell=\frac{n_\ell}{\Delta_h},
 \qquad
 n_\ell\in \F_q[X],
 \qquad
 \deg n_\ell\le d(h^2-\ell).
\end{equation}
Indeed, to obtain $n_\ell$ we replace column $\ell$ of the Hankel matrix by
$(-\mu_{i+h})_i$. Each term in the new determinant has degree at most $d$
times
\[
 \sum_i i+\sum_{k\ne\ell}k+h=h^2-\ell.
\]
Define
\[
 G(Z)=Z^h+\sum_{\ell=0}^{h-1}\gamma_\ell Z^\ell\in \F_q(X)[Z].
\]

At a maximum-support point $\xi$, the inequality
\(\Delta_h(\xi)\ne0\) means that every rational coefficient in
\eqref{eq:reconstruction-8} is defined after substituting \(X=\xi\).  Write $G_\xi$ for
the polynomial obtained by making this substitution in the coefficients of
$G$.  The specialized system is
\[
 \mathbf V_\xi\mathbf W_\xi
 \bigl(G_\xi(w_1),\ldots,G_\xi(w_h)\bigr)^{\mathsf T}=0.
\]
Both matrices are invertible, so $G_\xi(w_s)=0$ for every $s$.  As both sides
below are monic of degree $h$,
\begin{equation}\label{eq:reconstruction-11}
 G_\xi(Z)=\prod_{s=1}^h(Z-w_s).
\end{equation}
In particular, $G_\xi$ has distinct roots, so
$\operatorname{Disc}(G_\xi)\ne0$. All coefficients of $G$ are defined at
$\xi$. Therefore $\operatorname{Disc}(G)$ specializes to this nonzero
element and is itself nonzero in $\F_q(X)$. Let $\Omega$ be a splitting
field of $G$ and write its roots, with multiplicity, as
$\alpha_1,\ldots,\alpha_h$. By the root formula,
\[
 0\ne\operatorname{Disc}(G)
   =\prod_{1\le r<s\le h}(\alpha_s-\alpha_r)^2 .
\]
Thus $G(Z)=\prod_{s=1}^h(Z-\alpha_s)$ with pairwise distinct roots.

We imposed the recurrence in \eqref{eq:reconstruction-8} for the first $h$ shifts. To
prove it for every $0\le r\le T-h$, clear the common denominator
$\Delta_h$ and consider the polynomial
\[
 E_r=\Delta_h\mu_{r+h}+\sum_{\ell=0}^{h-1}n_\ell\mu_{r+\ell}\in \F_q[X].
\]
At every point $\xi$ with $h$ nodes, \eqref{eq:reconstruction-1} and \eqref{eq:reconstruction-11}
give $E_r(\xi)=\Delta_h(\xi)\sum_{w\in S(\xi)}b_{\xi,w}w^rG_\xi(w)=0$.
On the other hand, \eqref{eq:reconstruction-5} and \eqref{eq:reconstruction-9} give
\begin{equation}\label{eq:reconstruction-13}
 \deg E_r\le d(r+h^2)
 \le d(T-h+h^2)
 =d\bigl(T+h(h-1)\bigr).
\end{equation}
By \eqref{eq:reconstruction-7}, $E_r$ vanishes at more points than its degree. Thus $E_r=0$, and division by
$\Delta_h$ gives
\begin{equation}\label{eq:reconstruction-14}
 \mu_{r+h}+\sum_{\ell=0}^{h-1}\gamma_\ell\mu_{r+\ell}=0
 \qquad (0\le r\le T-h).
\end{equation}

It remains to recover the weights. The Vandermonde matrix of the roots,
\[
 \mathbf V_\alpha=(\alpha_s^j)_{0\le j<h,\,1\le s\le h},
\]
is invertible.  Define $\rho_1,\ldots,\rho_h\in\Omega$ uniquely by
\[
 \mathbf V_\alpha(\rho_1,\ldots,\rho_h)^{\mathsf T}
 = (\mu_0,\ldots,\mu_{h-1})^{\mathsf T}.
\]
The sequence $\sum_s\rho_s\alpha_s^j$ has the same first $h$ values as
$\mu_j$ and obeys the recurrence induced by $G$.  Indeed,
\(G(\alpha_s)=0\), multiplied by \(\alpha_s^r\), gives
\[
  \alpha_s^{r+h}
  +\sum_{\ell=0}^{h-1}\gamma_\ell\alpha_s^{r+\ell}=0.
\]
After multiplying by \(\rho_s\) and summing over \(s\), this is precisely
the recurrence in \eqref{eq:reconstruction-14}.  Equation
\eqref{eq:reconstruction-14} therefore proves \eqref{eq:reconstruction-3} by induction through
index $T$.

Finally, these identities through index $2h-2$ give a factorization of
the Hankel matrix over $\Omega$,
\[
 \bigl(\mu_{i+j}\bigr)_{i,j<h}
 =\mathbf V_\alpha\operatorname{diag}(\rho_1,\ldots,\rho_h)\mathbf V_\alpha^{\mathsf T}.
\]
Consequently,
\[
 \Delta_h=\det(\mathbf V_\alpha)^2\prod_{s=1}^h\rho_s.
\]
The left side and the Vandermonde determinant are nonzero, so every
$\rho_s$ is nonzero.
\end{proof}

To compare the reconstructed moments of two tables, we work
over $\F_q(X)$ even when the roots lie in different extensions. We can
do so by putting the nodes and their weights into a single rational
function of $Z$. In Liu's analysis of CVP
\cite{liu2026cvp}, all nodes have weight one, so summing
$1/(Z-\alpha_s)$ over the roots of $G$ gives $G'(Z)/G(Z)$.
Here each term has its own weight $\rho_s$. We must show that their sum
still has denominator $G$ in lowest terms.

\begin{corollary}[A Reduced Fraction for the Moments]
\label{cor:reduced-cauchy-fraction}
Under the conclusion of Lemma~\ref{lem:weighted-reconstruction}, put
\[
  G(Z)=\sum_{k=0}^{h}g_kZ^k
      =\prod_{s=1}^{h}(Z-\alpha_s),
  \qquad g_h=1.
\]
If $h\ge1$, define
\begin{equation}\label{eq:cauchy-numerator-root-form}
  A(Z)=\sum_{s=1}^{h}\rho_s\frac{G(Z)}{Z-\alpha_s}.
\end{equation}
Thus $A(Z)/G(Z)=\sum_{s=1}^h\rho_s/(Z-\alpha_s)$. We have $A\in \F_q(X)[Z]$, $\deg_ZA<h$, and $\gcd(A,G)=1$.  Moreover, if
\begin{equation}\label{eq:generated-moments}
  \widehat\mu_j:=\sum_{s=1}^{h}\rho_s\alpha_s^j
  \qquad(j\ge0),
\end{equation}
then $\widehat\mu_j=\mu_j$ for $0\le j\le T$ and
\begin{equation}\label{eq:cauchy-laurent-expansion}
  \frac{A(Z)}{G(Z)}
  =\sum_{j\ge0}\widehat\mu_jZ^{-j-1}
  \qquad\text{in }\F_q(X)((Z^{-1})).
\end{equation}
Here $\F_q(X)((Z^{-1}))$ is the field of formal Laurent series in
$Z^{-1}$ over $\F_q(X)$.
For $h=0$, the same notation means
\[
  (A,G)=(0,1),
  \qquad \widehat\mu_j=0\quad(j\ge0).
\]
\end{corollary}

\begin{proof}
The numerator has the following expression in terms of the moments over $\F_q(X)$:
\begin{equation}\label{eq:cauchy-numerator-moment-form}
  A(Z)=\sum_{k=1}^{h}g_k
       \sum_{r=0}^{k-1}\mu_rZ^{k-1-r}.
\end{equation}
To verify this identity, use $G(\alpha_s)=0$ and
\[
  \frac{Z^k-\alpha_s^k}{Z-\alpha_s}
  =\sum_{r=0}^{k-1}Z^{k-1-r}\alpha_s^r,
\]
multiply by $g_k$, and sum over $k$ and then over $s$ with weights
$\rho_s$. The highest moment index used is $h-1\le T$, so \eqref{eq:reconstruction-3}
applies. In particular \eqref{eq:cauchy-numerator-moment-form} gives
$A\in\F_q(X)[Z]$ and $\deg_Z A<h$.

At a root $\alpha_s$, every summand in
\eqref{eq:cauchy-numerator-root-form} except the $s$th vanishes, while the
$s$th has value
\begin{equation}\label{eq:cauchy-reducedness-at-root}
  A(\alpha_s)
  =\rho_s\prod_{t\ne s}(\alpha_s-\alpha_t)\ne0.
\end{equation}
This follows by canceling $Z-\alpha_s$ in $G(Z)$ and evaluating the
remaining factors at $\alpha_s$. All those factors and the weight $\rho_s$
are nonzero. A nonconstant common divisor of $A$ and $G$ would have some
$\alpha_s$ as a root, which contradicts
\eqref{eq:cauchy-reducedness-at-root}. Thus $\gcd(A,G)=1$.

Finally, as a formal series at infinity,
\[
  \frac{1}{Z-\alpha_s}
  =\sum_{j\ge0}\alpha_s^jZ^{-j-1}.
\]
Summing these expansions proves
\eqref{eq:cauchy-laurent-expansion}.  The reconstruction identities give
$\widehat\mu_j=\mu_j$ for every index $0\le j\le T$ supplied by the code.
Formula \eqref{eq:generated-moments} extends this sequence to all $j>T$
using the same recurrence. When $h=0$ the fraction and all these
coefficients are zero.
\end{proof}

\begin{remark}[The Number of Evaluation Points]
The bound in \eqref{eq:reconstruction-2} accounts for two losses.
Up to $dh(h-1)$ points may have fewer
than $h$ nodes. At each of the remaining points $E_r$ vanishes, and
its degree is at most $d(T+h(h-1))$ by \eqref{eq:reconstruction-13}. Thus
$|P'|>d(T+2h(h-1))$ suffices for the actual support bound $h$, and
\eqref{eq:reconstruction-2} ensures this for every $h\le\kappa$. Equality between the
number of zeros and the degree would leave the root-counting argument
inconclusive. For instance, $X^{q-1}-1$ vanishes at every point of
$\F_q^\times$.
\end{remark}

We can now compare the fraction $A/G$ for different tables in the common
field $\F_q(X)(Z)$. To use the clause constraints in this comparison, we
still have to show how a support polynomial records the Boolean values at
the anchors.

\section{Recovering the Boolean Values at the Anchors}
\label{sec:anchors}

Fix a subtype $\tau=(C,\beta)$ and a variable $i\in I_C$.
In an honest codeword, this subtype is used precisely when the satisfying
assignment restricts to $\beta$ on $C$.  Its encoding polynomial then takes
the value $\beta_i$ at $a_i$. For an arbitrary codeword to which
Lemma~\ref{lem:weighted-reconstruction} applies, we obtain distinct nodes
$\alpha_1,\ldots,\alpha_h$ in an extension of $\F_q(X)$.  We would like to
set $X=a_i$, but we only know that the nodes are algebraic over $\F_q(X)$.
Under the moment bounds below, we will show that we can evaluate the
coefficients of their support
polynomial $G_\tau(Z)=\prod_{s=1}^h(Z-\alpha_s)$ at $X=a_i$, and that
doing so gives $(Z-\beta_i)^h$.  In particular, two subtypes that disagree
on the value of $x_i$ have disjoint sets of nodes.

The CVP reduction of OpenAI~\cite[Chapter~7, Lemma~12]{openai-ten-advances}
recovers this Boolean information by studying the roots in a splitting
field.  For our weighted moments we work with the coefficients of
$G_\tau$ directly.  The first step is to understand what happens to the
nodes when we pass from the ordinary moments to the shifted moments.

\subsection{The Shifted Moments}

Let $P_\tau\subseteq P$ be a set of evaluation points at which the fibers
of type $\tau$ have support at most $\kappa$.  Suppose that
Lemma~\ref{lem:weighted-reconstruction} applies to this set, and write
\begin{equation}\label{eq:anchor-ordinary-moments}
  \mu_{\tau,j}=\sum_{s=1}^h\rho_s\alpha_s^j
  \qquad(0\le j\le T)
\end{equation}
in a splitting field $\Omega$ of $G_\tau$.  The weights $\rho_s$ are
nonzero.  As before, if $h=0$ then $G_\tau=1$.

Put $t=X-a_i$.  At every $\xi\in P_\tau$, the ordinary and shifted
moments use the same weights.  Multiplying the $j$th shifted moment by
$(\xi-a_i)^j$ leaves the weighted sum of $(w-\beta_i)^j$.  Expanding
by the binomial theorem gives
\begin{equation}\label{eq:shifted-binomial-transfer}
 (\xi-a_i)^j\eta_{\tau,i,j}(\xi)
 =\sum_{r=0}^{j}\binom jr(-\beta_i)^{j-r}\mu_{\tau,r}(\xi).
\end{equation}
As polynomials in $X$, both sides have degree at most $dj$.  On the left
this follows from $j+(d-1)j=dj$, and on the right from $dr\le dj$.
Thus, provided that
\begin{equation}\label{eq:shifted-transfer-root-count}
  dT<|P_\tau|,
\end{equation}
the pointwise equalities give a polynomial identity for every $0\le j\le T$:
\begin{equation}\label{eq:shifted-binomial-polynomial}
 t^j\eta_{\tau,i,j}
 =\sum_{r=0}^{j}\binom jr(-\beta_i)^{j-r}\mu_{\tau,r}.
\end{equation}
Substituting~\eqref{eq:anchor-ordinary-moments} on the right, we get
$\sum_s\rho_s(\alpha_s-\beta_i)^j$.  Dividing by $t^j$ in $\Omega$ gives
\begin{equation}\label{eq:transformed-moments}
  \eta_{\tau,i,j}=\sum_{s=1}^{h}\rho_s y_s^j,
  \qquad
  y_s=\frac{\alpha_s-\beta_i}{t}.
\end{equation}
The shift preserves all the weights and keeps the nodes distinct.  Its
advantage is that the $j$th moment of the new nodes has degree at most
$(d-1)j$, by the code constraint~\eqref{eq:shifted-rs}.

\subsection{A Polynomial for the Shifted Nodes}

We next show that the polynomial with roots $y_1,\ldots,y_h$ has
coefficients in $\F_q[X]$.  For this, we will group $h$ consecutive
moments into a vector and study the linear map that shifts the vector
forward by one moment.

\begin{lemma}[A Polynomial for the Shifted Nodes]
\label{lem:polynomial-state-stabilization}
Let $d,h\ge1$.  In a field $\Omega$ containing $\F_q(X)$, let
$y_1,\ldots,y_h$ be distinct and let $\rho_1,\ldots,\rho_h$ be nonzero.
Suppose
\begin{equation}\label{eq:state-module-hypothesis}
  \eta_j=\sum_{s=1}^{h}\rho_s y_s^j\in \F_q[X],
  \qquad
  \deg\eta_j\le(d-1)j
  \qquad(0\le j\le T).
\end{equation}
If
\begin{equation}\label{eq:state-module-exact-budget}
  T\ge(d-1)h(h-1)+2h-1,
\end{equation}
then
\begin{equation}\label{eq:state-module-conclusion}
  H(Y):=\prod_{s=1}^{h}(Y-y_s)\in \F_q[X,Y].
\end{equation}
In particular, it suffices to take $T\ge dh^2$.
\end{lemma}

\begin{proof}
For $0\le j\le T-h+1$, let
\[
  \mathbf v_j=(\eta_j,\eta_{j+1},\ldots,\eta_{j+h-1})^{\mathsf T}
  \in\F_q[X]^h.
\]
We first construct a linear map $\mathbf S$ over $\F_q(X)$ whose
characteristic polynomial is $H$.  Let
$\mathbf V=(y_s^r)_{0\le r<h,\,1\le s\le h}$ be the Vandermonde matrix
on the nodes.  The first $h$ vectors form the Hankel matrix
\begin{equation}\label{eq:shifted-hankel-factorization}
 \mathbf H_0=[\mathbf v_0\ \cdots\ \mathbf v_{h-1}]
 =\mathbf V\operatorname{diag}(\rho_1,\ldots,\rho_h)\mathbf V^{\mathsf T}.
\end{equation}
The nodes are distinct and the weights nonzero, so
\begin{equation}\label{eq:shifted-hankel-determinant}
 \det\mathbf H_0=\det(\mathbf V)^2\prod_{s=1}^{h}\rho_s\ne0.
\end{equation}
Thus $\mathbf v_0,\ldots,\mathbf v_{h-1}$ form a basis of $\F_q(X)^h$.
Define $\mathbf S$ by $\mathbf S\mathbf v_j=\mathbf v_{j+1}$ on this
basis.  The vector $\mathbf v_h$ exists since $T\ge2h-1$.
Over $\Omega$, we have $\mathbf v_j=\mathbf V(\rho_s y_s^j)_{s=1}^h$,
where the weighted powers on the right form a column vector.  Hence, in these
coordinates, shifting forward multiplies the $s$th entry by $y_s$.  This gives
\begin{equation}\label{eq:state-shift-diagonalization}
  \mathbf S=\mathbf V\operatorname{diag}(y_1,\ldots,y_h)\mathbf V^{-1}.
\end{equation}
It follows that $\mathbf S\mathbf v_j=\mathbf v_{j+1}$ for every
$0\le j\le T-h$ and that
\begin{equation}\label{eq:state-char-poly}
  \det(Y\mathbf I-\mathbf S)=\prod_{s=1}^{h}(Y-y_s)=H(Y).
\end{equation}

It remains to find a basis in which $\mathbf S$ has polynomial entries.
For $h-1\le k\le T-h+1$, let $\mathcal M_k$ consist of all polynomial
linear combinations of $\mathbf v_0,\ldots,\mathbf v_k$,
\[
  \mathcal M_k=\F_q[X]\mathbf v_0+\cdots+\F_q[X]\mathbf v_k
  \subseteq\F_q[X]^h.
\]
Each $\mathcal M_k$ has a basis over $\F_q[X]$, meaning that every element is
a unique polynomial linear combination of the basis vectors. Such a basis
exists for any polynomial span in $\F_q[X]^h$~\cite{stacks-pid}. Each basis has
$h$ elements, since the first $h$ vectors already span $\F_q(X)^h$.
If $\mathcal M_k=\mathcal M_{k+1}$ for some $k$, then $\mathbf S$ sends
every generator of $\mathcal M_k$ into $\mathcal M_k$.  Its matrix in a
basis of this module will therefore have polynomial entries.

We bound the number of strict inclusions by the degree of a determinant.
Let $b=\deg\det\mathbf H_0$.
The determinant term indexed by a permutation $\pi$ has degree at most
\begin{equation}\label{eq:state-determinant-degree}
  (d-1)\sum_{r=0}^{h-1}(r+\pi(r))=(d-1)h(h-1),
\end{equation}
so $b\le(d-1)h(h-1)$.  For each $k$ choose a matrix $\mathbf B_k$ whose
columns are a polynomial basis of $\mathcal M_k$, with
$\mathbf B_{h-1}=\mathbf H_0$.  Changing such a basis multiplies its
determinant by a nonzero constant.  Since $\mathcal M_k\subseteq
\mathcal M_{k+1}$, we have $\mathbf B_k=\mathbf B_{k+1}\mathbf A$
for a polynomial matrix $\mathbf A$.  If $\det\mathbf A$ is a nonzero
constant, then $\mathbf A^{-1}$ has polynomial entries and the two modules
are equal.  Thus for each strict inclusion we have
\begin{equation}\label{eq:determinant-index-drop}
  \deg\det\mathbf B_{k+1}\le\deg\det\mathbf B_k-1.
\end{equation}

By~\eqref{eq:state-module-exact-budget}, $h+b\le T-h+1$, so we can
form all the modules in the chain
\[
  \mathcal M_{h-1}\subseteq\mathcal M_h\subseteq\cdots
  \subseteq\mathcal M_{h+b}.
\]
There are $b+1$ inclusions.  The determinant degree starts at $b$ and
stays nonnegative, so at least one inclusion is an equality.
In the corresponding basis, $\mathbf S$ has entries in $\F_q[X]$.
Thus its characteristic polynomial~\eqref{eq:state-char-poly} lies
in $\F_q[X,Y]$ as claimed.

Finally, $T\ge dh^2$ implies~\eqref{eq:state-module-exact-budget} since
\[
 dh^2-\bigl((d-1)h(h-1)+2h-1\bigr)=h^2+(d-3)h+1\ge0.
\]
For $d=1$ the right side is $(h-1)^2$, and increasing $d$ only increases it.
\end{proof}

\subsection{Evaluating the Support Polynomial at an Anchor}

\begin{lemma}[The Value at an Anchor]
\label{lem:anchor-support-congruence}
Fix a codeword of $\cC_\varphi$, a subtype $\tau=(C,\beta)$, and $i\in I_C$.
Suppose that weighted reconstruction
on $P_\tau$ gives a support of size $h\le\kappa$ for $\tau$ and that the
ordinary and shifted moments use the same coefficients in each retained
fiber.  If
\begin{equation}\label{eq:anchor-congruence-hypotheses}
  dT<|P_\tau|,
  \qquad T\ge d\kappa^2,
\end{equation}
then
\begin{equation}\label{eq:anchor-support-polynomial}
  G_\tau(Z)=\prod_{s=1}^h(Z-\alpha_s)\in\F_q[X,Z],
\end{equation}
and
\begin{equation}\label{eq:anchor-support-congruence}
  G_\tau(Z)\equiv(Z-\beta_i)^h\pmod{X-a_i}.
\end{equation}
For $h=0$ we have $G_\tau=1$.
\end{lemma}

\begin{proof}
The claim holds for the empty support, so assume $h\ge1$.
Since $dT<|P_\tau|$, identity~\eqref{eq:transformed-moments} expresses
the shifted moments in terms of the shifted nodes.  The $j$th shifted moment has degree at most $(d-1)j$, and $T\ge d\kappa^2\ge dh^2$.  We may therefore apply
Lemma~\ref{lem:polynomial-state-stabilization} to get
\[
  H_i(Y)=\prod_{s=1}^h\left(Y-\frac{\alpha_s-\beta_i}{X-a_i}\right)
  \in\F_q[X,Y].
\]
To recover $G_\tau$ from $H_i$, put $t=X-a_i$.  Then
\begin{equation}\label{eq:G-from-H}
  G_\tau(Z)=t^h H_i\!\left(\frac{Z-\beta_i}{t}\right).
\end{equation}
Writing $H_i(Y)=Y^h+\sum_{r=0}^{h-1}b_rY^r$, we obtain
\[
  G_\tau(Z)=(Z-\beta_i)^h
    +\sum_{r=0}^{h-1}b_r t^{h-r}(Z-\beta_i)^r.
\]
All the coefficients on the right are polynomials in $X$, and setting
$X=a_i$ leaves only $(Z-\beta_i)^h$.
\end{proof}

\begin{corollary}[Conflicting Assignments Have Disjoint Roots]
\label{cor:anchor-incompatibility}
Let $\tau=(C,\beta)$ and $\tau'=(C',\beta')$ be reconstructed subtypes
satisfying the hypotheses of Lemma~\ref{lem:anchor-support-congruence}
at a shared variable $i\in I_C\cap I_{C'}$.  If $\beta_i\ne\beta'_i$, then
\[
  \gcd_{\F_q(X)[Z]}(G_\tau,G_{\tau'})=1.
\]
In particular their root sets in an algebraic closure of $\F_q(X)$ are
disjoint.
\end{corollary}

\begin{proof}
We can assume that both degrees $h,h'$ are positive, since $1$ is coprime
to every polynomial.  Recall that the resultant of two monic polynomials
is a polynomial in their coefficients, and vanishes exactly when they
have a common root in an algebraic closure~\cite{cox-little-oshea-2015}.
We can thus compute the resultant at $X=a_i$ by first setting $X=a_i$ in
the two polynomials.  Lemma~\ref{lem:anchor-support-congruence} gives
\begin{align*}
 \operatorname{Res}_Z(G_\tau,G_{\tau'})\bmod(X-a_i)
 &=\operatorname{Res}_Z((Z-\beta_i)^h,(Z-\beta'_i)^{h'})\\
 &=(\beta_i-\beta'_i)^{hh'}\ne0,
\end{align*}
since $\beta_i$ and $\beta'_i$ are distinct Boolean values in $\F_q$.
The resultant is therefore nonzero in $\F_q[X]$, and the polynomials have
no common root over $\F_q(X)$ or any extension of it.
\end{proof}

\section{From Local Roots to a Global Assignment}\label{sec:decoding}

It remains to see how the different tables constrain one another.
For a codeword $\mathbf u\in\cC_\varphi$, write
\[
  S_\tau(\xi)=\{w\in\F_q:\mathbf u_{\tau,\xi,w}\ne0\}
\]
for its support in the fiber $(\tau,\xi)$.  The following proposition
concerns codewords of total weight at most
\begin{equation}\label{eq:sparse-threshold}
  S_0=4M^\delta R.
\end{equation}
It shows that, for an unsatisfiable formula, even one nonempty fiber
with at most $T+1$ nodes rules out such a codeword.

\begin{proposition}[Support in a Nonempty Fiber]\label{prop:sparse-dichotomy}
Suppose $\varphi$ is unsatisfiable.  If
$0\ne\mathbf u\in\cC_\varphi$ and $\wt(\mathbf u)\le S_0$, then every
nonempty evaluation fiber satisfies
\begin{equation}\label{eq:heavy-fiber-conclusion}
  |S_\tau(\xi)|\ge T+2.
\end{equation}
\end{proposition}

\begin{proof}
\emph{Reconstructing the Tables.}
For each type $\tau$, retain the points
\begin{equation}\label{eq:good-points}
  P_\tau=\{\xi\in P:|S_\tau(\xi)|\le\kappa\}.
\end{equation}
Every discarded fiber contributes more than $\kappa$ coordinates to the
total support.  By~\eqref{eq:R-over-g} we have $M<40S^D$, and
by~\eqref{eq:M-g-R} we have $R=(m+1)|P|<Sq$.  Using
$D\delta+1-c=-1$, we get
\begin{equation}\label{eq:discarded-points}
  |P\setminus P_\tau|
  \le\frac{4M^\delta R}{\kappa}
  <4\cdot40^\delta q/S<q/10.
\end{equation}
For the last inequality, $\delta=(c-2)/(4c+7)<1/4$, so
$40^\delta<40^{1/4}<51/20$.  Since $S\ge102$,
$4\cdot40^\delta/S<4(51/20)/102=1/10$.

We now check that enough points remain to apply
Lemma~\ref{lem:weighted-reconstruction}.  Since $|P|=q-n$ and $n<q/16$,
while $d\le S$, $T=4S^{2c+1}$, $\kappa=S^c$, and $q\ge S^{2c+3}$, we have
\begin{align}
 |P_\tau|&>q-n-q/10>67q/80,
   \label{eq:reconstruction-left-margin}\\
 d\bigl(T+2\kappa(\kappa-1)\bigr)
   &<q\left(\frac4S+\frac{2}{S^2}\right)<q/25.
   \label{eq:reconstruction-right-margin}
\end{align}
Indeed, $4/S\le 2/51$ and $2/S^2<1/1275$, whose upper bounds sum to
$1/25$.  In particular, $|P_\tau|>d(T+2\kappa(\kappa-1))$, and we also have
\begin{equation}\label{eq:reconstruction-other-margin}
   dT<|P_\tau|,\qquad T\ge2\kappa-1.
\end{equation}

Apply Lemma~\ref{lem:weighted-reconstruction} to each type $\tau$ on its
own set $P_\tau$.  If the reconstructed support has size $h_\tau\le\kappa$,
Corollary~\ref{cor:reduced-cauchy-fraction} gives the reduced fraction
\begin{equation}\label{eq:type-cauchy-data}
  \frac{A_\tau(Z)}{G_\tau(Z)},\qquad
  \deg_Z A_\tau<h_\tau=\deg_Z G_\tau,\qquad \gcd(A_\tau,G_\tau)=1.
\end{equation}
Here $G_\tau$ is monic.  Expanding the fraction at infinity, the coefficients
of $Z^{-1},\ldots,Z^{-T-1}$ are $\mu_{\tau,0},\ldots,\mu_{\tau,T}$.  The empty support gives the fraction $0/1$.
Although the nodes for different types may lie in different splitting
fields, all these fractions lie in $\F_q(X)(Z)$.

For a subtype $\tau=(C,\beta)$, the ordinary and shifted moments come
from the same coefficients on $P_\tau$.  We have checked that
$dT<|P_\tau|$, and $T=4S^{2c+1}\ge d\kappa^2$ by~\eqref{eq:moment-budget}.
Hence Lemma~\ref{lem:anchor-support-congruence} gives
\begin{equation}\label{eq:subtype-support-congruence}
  G_{(C,\beta)}(Z)\equiv(Z-\beta_i)^{h_{(C,\beta)}}\pmod{X-a_i}
  \qquad(i\in I_C).
\end{equation}

\emph{Factoring the Global Support Polynomial.}
Fix a clause $C$.  We want to show that its subtype support polynomials
multiply to $G_\star$.  By the clause constraint~\eqref{eq:clause-decomposition},
the global moment at each evaluation point is the sum of the corresponding
subtype moments.  Uniqueness of the representing polynomials of degree at
most $dj$ therefore gives
\begin{equation}\label{eq:clause-moment-equality}
  \mu_{\star,j}=\sum_{\beta\in\mathrm{Sat}(C)}\mu_{(C,\beta),j}
  \qquad(0\le j\le T).
\end{equation}

First, these $T+1$ coefficients suffice to prove equality of the fractions.
Write their difference over the common denominator
\[
  G_\star(Z)\prod_{\beta\in\mathrm{Sat}(C)}G_{(C,\beta)}(Z).
\]
Its degree is
\begin{equation}\label{eq:clause-total-denominator-degree}
  m_C=h_\star+\sum_{\beta\in\mathrm{Sat}(C)}h_{(C,\beta)}
   \le (1+|\mathrm{Sat}(C)|)\kappa\le9\kappa<T+1.
\end{equation}
If $m_C=0$, all the fractions are zero.  Otherwise, the numerator of the
difference has degree less than $m_C$ since each summand is a proper
fraction.  If this numerator were nonzero of degree $e$, the leading
term of the difference at infinity would be a nonzero multiple of
$Z^{-(m_C-e)}$.  But~\eqref{eq:clause-moment-equality} tells us that
all coefficients from $Z^{-1}$ through $Z^{-T-1}$ vanish.  We conclude that
\begin{equation}\label{eq:clause-cauchy-identity}
  \frac{A_\star}{G_\star}
   =\sum_{\beta\in\mathrm{Sat}(C)}
     \frac{A_{(C,\beta)}}{G_{(C,\beta)}}.
\end{equation}

Second, no factor of a subtype denominator cancels on the right.  Indeed,
two distinct assignments to $C$ differ on some variable.  By
Corollary~\ref{cor:anchor-incompatibility} the subtype denominators
are therefore pairwise coprime.  Write the sum on the right over the
common denominator $G_C$ with numerator $A_C$, where
\[
  G_C=\prod_{\beta\in\mathrm{Sat}(C)}G_{(C,\beta)},\qquad
  A_C=\sum_{\beta\in\mathrm{Sat}(C)}A_{(C,\beta)}
       \frac{G_C}{G_{(C,\beta)}}.
\]
Let $f\in\F_q(X)[Z]$ be an irreducible factor of $G_{(C,\beta)}$.
Modulo $f$, only the $\beta$th summand of $A_C$ can survive.  It does
survive, since $f$ divides neither $A_{(C,\beta)}$ (the fraction is
reduced) nor $G_C/G_{(C,\beta)}$ (the denominators are pairwise coprime).
It follows that
\begin{equation}\label{eq:no-denominator-cancellation}
  \gcd(A_C,G_C)=1.
\end{equation}
Now~\eqref{eq:clause-cauchy-identity} gives $A_\star G_C=A_CG_\star$.
Coprimality on the two sides implies both $G_\star\mid G_C$ and
$G_C\mid G_\star$.  As both denominators are monic, we have proved
\begin{equation}\label{eq:support-polynomial-factorization}
  G_\star=\prod_{\beta\in\mathrm{Sat}(C)}G_{(C,\beta)}
  \qquad\text{for every clause }C.
\end{equation}

\emph{Choosing an Assignment.}
Suppose $G_\star\ne1$.  Pick any monic irreducible factor $f$ of
$G_\star$ over $\F_q(X)$.  For each clause $C$, the factorization
in~\eqref{eq:support-polynomial-factorization} puts $f$ in at least one
subtype denominator and pairwise coprimality puts it in exactly one.
Let $\beta(C)$ be the satisfying assignment indexing that subtype.
If two chosen assignments disagreed on a shared variable, their
support polynomials would be coprime by
Corollary~\ref{cor:anchor-incompatibility}, yet both contain $f$.
Thus the local assignments agree on every shared variable.  Every
remaining variable occurs in a clause, since we removed unused variables
at the start of the reduction.  The local assignments therefore give
a single assignment that satisfies all the clauses, a contradiction.

It follows that $G_\star=1$.  Since the factors
in~\eqref{eq:support-polynomial-factorization} are monic, each subtype
denominator is also $1$.  Thus all reconstructed supports are empty,
and all their moments through order $T$ vanish:
\begin{equation}\label{eq:all-generic-moments-zero}
  \mu_{\tau,j}=0\qquad(\tau\in\Theta,\ 0\le j\le T).
\end{equation}

Finally, return to any evaluation fiber, including one that we discarded.
Suppose its support consists of $1\le k\le T+1$ distinct nodes
$w_1,\ldots,w_k$ with nonzero weights $b_1,\ldots,b_k$.  Since
\eqref{eq:all-generic-moments-zero} holds as a polynomial identity,
evaluating the first $k$ moments at this fiber's point $\xi$ gives
\[
  0=\mu_{\tau,j}(\xi)=\sum_{s=1}^k b_sw_s^j
  \qquad(0\le j<k).
\]
The coefficient matrix $(w_s^j)_{0\le j<k,\,1\le s\le k}$ is Vandermonde
on distinct nodes, so it is invertible.  All $k$ weights must be zero, a
contradiction.  Hence every nonempty fiber has at least $T+2$ nodes.
\end{proof}

\section{From the Code to a Lattice}\label{sec:lattice}

We now turn the two properties of $\cC_\varphi$ into a gap for SVP.
First we take all integer vectors that reduce to codewords modulo $q$.
This gives a lattice, but we still have to distinguish a vector with a
single $1$ in each of its nonempty fibers from one with many nonzero entries
in a fiber. For this we apply a Hadamard matrix separately to each fiber,
as in our earlier reduction~\cite{hair-sahai-2511}.

\subsection{A Basis for the Construction A Lattice}

The integer vectors whose residues lie in our code form the standard
Construction A lattice~\cite{conway-sloane-1999}
\begin{equation}\label{eq:construction-a-lattice}
  \Lambda_\varphi
  =\{\mathbf z\in\Z^M:\mathbf z\bmod q\in\cC_\varphi\}.
\end{equation}
To output a basis, put $k=\dim_{\F_q}\cC_\varphi$ and use Gaussian
elimination to find $k$ pivot columns of a generator matrix. Let
$\boldsymbol\Pi$ be the permutation matrix that moves these columns to the
front. In this order the row-reduced generator is $(\mathbf I_k\mid\mathbf P)$
for some $\mathbf P\in\F_q^{k\times(M-k)}$, so
\[
  \cC_\varphi\boldsymbol\Pi
  =\{(\mathbf f,\mathbf f\mathbf P):\mathbf f\in\F_q^k\}.
\]
Let $\widetilde{\mathbf P}$ be the integer matrix obtained by taking the
representative in $\{0,\ldots,q-1\}$ of each entry of $\mathbf P$. Set
\begin{equation}\label{eq:construction-a-basis}
  \mathbf B'=
  \begin{pmatrix}
    \mathbf I_k&\widetilde{\mathbf P}\\
    0&q\mathbf I_{M-k}
  \end{pmatrix},
  \qquad
  \mathbf B_A=\mathbf B'\boldsymbol\Pi^{\mathsf T}.
\end{equation}
The last permutation puts the columns back in their original order so that
the $q$ coordinates of each evaluation fiber are again consecutive.

\begin{lemma}[A Basis for $\Lambda_\varphi$]\label{lem:construction-a-basis}
The rows of $\mathbf B_A$ generate $\Lambda_\varphi$, and
$|\det\mathbf B_A|=q^{M-k}$.
\end{lemma}

\begin{proof}
In the permuted coordinates, an integer combination of the rows of
$\mathbf B'$ has the form
\[
  (\mathbf u,\mathbf v)\mathbf B'
    =(\mathbf u,\mathbf u\widetilde{\mathbf P}+q\mathbf v).
\]
Its residue belongs to $\cC_\varphi\boldsymbol\Pi$. Conversely, if
$(\mathbf a,\mathbf b)\in\Z^k\times\Z^{M-k}$ has residue in that code, then
$\mathbf b-\mathbf a\widetilde{\mathbf P}$ is divisible by $q$ in every
coordinate. Taking $\mathbf u=\mathbf a$ and
$\mathbf v=(\mathbf b-\mathbf a\widetilde{\mathbf P})/q$ gives the desired
integer combination. Multiplying by $\boldsymbol\Pi^{\mathsf T}$ now
gives precisely $\Lambda_\varphi$. The determinant follows from the
diagonal blocks of $\mathbf B'$ and $|\det\boldsymbol\Pi|=1$.
\end{proof}

\subsection{Applying a Hadamard Matrix to Each Fiber}

We need $q$ mutually orthogonal sign vectors, one for each coordinate of a
fiber. Take $L$ to be the least power of two at least $q$, so $q\le L<2q$.
The Sylvester Hadamard matrix~\cite{sylvester1867thoughts}
$\mathbf H_L\in\{\pm1\}^{L\times L}$ satisfies
$\mathbf H_L\mathbf H_L^{\mathsf T}=L\mathbf I_L$. In particular, the
matrix $\mathbf U$ formed by its first $q$ rows satisfies
\begin{equation}\label{eq:partial-hadamard}
  \mathbf U\mathbf U^{\mathsf T}=L\mathbf I_q.
\end{equation}
Let $\mathbf E$ be the block-diagonal matrix with one copy of $\mathbf U$
per evaluation fiber. The output basis $\mathbf B$ is given by
\begin{equation}\label{eq:embedding-and-B}
  \mathbf E=\mathbf I_g\otimes\mathbf U\in\Z^{M\times gL},
  \qquad
  \mathbf B=\mathbf B_A\mathbf E.
\end{equation}
The map $\mathbf v\mapsto\mathbf v\mathbf U$ is injective since
$(\mathbf v\mathbf U)\mathbf U^{\mathsf T}=L\mathbf v$. Hence
$\mathbf z\mapsto\mathbf z\mathbf E$ is injective as well. Recall that
$M=gq$. The rank and ambient dimension $N$ of our output lattice satisfy
\begin{equation}\label{eq:rank-and-ambient}
  \operatorname{rank}(\mathbf B)=M,
  \qquad M\le N=gL<2M.
\end{equation}

\begin{lemma}[Norm Bounds for the Hadamard Map]\label{lem:hadamard-energy}
For every $\mathbf v\in\R^q$ and every $2<p\le\infty$,
\begin{equation}\label{eq:one-block-hadamard}
  \|\mathbf v\mathbf U\|_p\ge L^{1/p}\|\mathbf v\|_2,
\end{equation}
where $1/\infty=0$. For every
$\mathbf z=(\mathbf z^{(1)},\ldots,\mathbf z^{(g)})\in\R^M$, written
as $g$ blocks of $q$ coordinates each,
\begin{equation}\label{eq:all-block-hadamard}
  \|\mathbf z\mathbf E\|_p
  \ge L^{1/p}g^{1/p-1/2}\|\mathbf z\|_2.
\end{equation}
\end{lemma}

\begin{proof}
By \eqref{eq:partial-hadamard},
$\|\mathbf v\mathbf U\|_2=\sqrt L\|\mathbf v\|_2$. Applying
Lemma~\ref{lem:norm-comparison} in dimension $L$ gives
\[
  \|\mathbf v\mathbf U\|_p
  \ge L^{1/p-1/2}\|\mathbf v\mathbf U\|_2
  =L^{1/p}\|\mathbf v\|_2.
\]
For finite $p$, sum the $p$th powers of this inequality over all $g$ blocks
and take the $p$th root. It remains to apply
Lemma~\ref{lem:norm-comparison} to the vector of block lengths
$(\|\mathbf z^{(b)}\|_2)_{b=1}^g$, whose Euclidean norm is
$\|\mathbf z\|_2$. For $p=\infty$, we use the largest block to get
\[
  \|\mathbf z\mathbf E\|_\infty
  \ge \max_b \|\mathbf z^{(b)}\|_2\ge g^{-1/2}\|\mathbf z\|_2.
\]
\end{proof}

By Lemma~\ref{lem:construction-a-basis}, the vectors in the output lattice are
exactly
\begin{equation}\label{eq:final-lattice-description}
  \cL(\mathbf B)=\{\mathbf z\mathbf E:\mathbf z\in\Lambda_\varphi\}.
\end{equation}

\section{Hardness for a Fixed Rational Norm}\label{sec:core-theorem}

\begin{theorem}[Hardness for Rational $p>2$]\label{thm:rational-core}
Fix a rational $p>2$ and an integer $c\ge4$ for which $\chi>0$ in
\eqref{eq:parameter-family}. With the parameters in
\eqref{eq:parameter-family}--\eqref{eq:q-range}, the construction of
Sections~\ref{sec:code}--\ref{sec:lattice} gives a deterministic
polynomial-time Karp reduction from $\ThreeSAT$ to
\[
  M^{\eps_c(p)}\text{-}\GapSVP_p,
  \qquad
  \eps_c(p)=\frac\chi D=\frac\chi{4c+7}.
\]
The radius $r$ is an integer. On a nontrivial input the output basis has
rank $M$ and ambient dimension $N\in[M,2M)$. A YES instance has a nonzero
vector of norm at most $r$. In a NO instance every nonzero vector has norm
strictly greater than $M^{\eps_c(p)}r$.
\end{theorem}

\subsection{Completeness}\label{sec:completeness}

Fix a satisfying assignment $\sigma$. The $0$--$1$ lift of its codeword
$\mathbf x^\sigma$ lies in $\Lambda_\varphi$ by
Lemma~\ref{lem:code-completeness}. Thus $\mathbf x^\sigma\mathbf E$ belongs to
the output lattice. Each of the $R$ nonempty fibers maps to one Hadamard
row and all remaining fibers map to zero, giving
\begin{equation}\label{eq:honest-norm}
  \|\mathbf x^\sigma\mathbf E\|_p^p=RL.
\end{equation}
Thus the length to which we compare NO instances is
\begin{equation}\label{eq:algebraic-radius}
  r_0=(RL)^{1/p}.
\end{equation}
For our output threshold we round this length up to an integer. Write the
fixed rational $p=a/b$ in lowest terms, with $a,b$ positive, and set
\begin{equation}\label{eq:integer-radius}
  r=\min\{j\in\Z_{\ge1}:j^a\ge(RL)^b\}.
\end{equation}
We can compute $r$ exactly by binary search using integer powers. Since
$r=\lceil r_0\rceil$ and $r_0\ge1$,
\begin{equation}\label{eq:radius-rounding}
  r_0\le r<r_0+1\le2r_0.
\end{equation}
In particular, $\lambda_1^{(p)}(\cL(\mathbf B))\le r$.

\subsection{Soundness}\label{sec:soundness}

Suppose $\varphi$ is unsatisfiable and let $\mathbf z\mathbf E$ be any
nonzero vector of $\cL(\mathbf B)$. By \eqref{eq:final-lattice-description},
$\mathbf z\in\Lambda_\varphi$ is nonzero, and its residue
$\overline{\mathbf z}=\mathbf z\bmod q$ is a codeword. We consider three
cases according to the Hamming weight of this codeword.

\paragraph{Case 1. The Residue Is Zero.}
Here $\mathbf z=q\mathbf w$ for some nonzero $\mathbf w\in\Z^M$. At least
one block of $\mathbf z$ therefore has Euclidean norm at least $q$, so
\eqref{eq:one-block-hadamard} gives
\[
  \|\mathbf z\mathbf E\|_p\ge qL^{1/p}.
\]
Since $R<Sq$ and $q\ge S^Q$, we obtain
\begin{equation}\label{eq:zero-residue-ratio}
  \frac{\|\mathbf z\mathbf E\|_p}{r_0}
  \ge\frac q{R^{1/p}}\ge S^{c_0},
  \qquad c_0=Q-\frac{Q+1}{p}.
\end{equation}

\paragraph{Case 2. The Residue Has Large Support.}
If $\wt(\overline{\mathbf z})>S_0$, each of the nonzero residue coordinates
comes from an integer of magnitude at least one. In particular,
\[
  \|\mathbf z\|_2^2\ge\wt(\overline{\mathbf z})>S_0=4M^\delta R.
\]
Applying \eqref{eq:all-block-hadamard} to the whole vector gives
\begin{align}
  \frac{\|\mathbf z\mathbf E\|_p}{r_0}
  &>
  \frac{L^{1/p}g^{1/p-1/2}\,2M^{\delta/2}R^{1/2}}
       {L^{1/p}R^{1/p}} \notag\\
  &=2M^{\delta/2}\left(\frac Rg\right)^{1/2-1/p}
   \ge 2\,8^{-\vartheta}M^{(c-2)/(2D)}.
  \label{eq:large-support-ratio}
\end{align}
Here $R/g\ge1/8$ by \eqref{eq:R-over-g}, $\vartheta=1/2-1/p$, and
$\delta=(c-2)/D$.

\paragraph{Case 3. The Residue Has Small Nonempty Support.}
Suppose $0<\wt(\overline{\mathbf z})\le S_0$. By
Proposition~\ref{prop:sparse-dichotomy}, each nonempty fiber of the
residue has at least $T+2$ nonzero entries. Choose one such fiber. Its
integer lift has Euclidean norm at least $\sqrt{T+2}$, so
\eqref{eq:one-block-hadamard} gives
\[
  \|\mathbf z\mathbf E\|_p\ge L^{1/p}\sqrt{T+2}.
\]
Using $\sqrt{T+2}>\sqrt T=2S^{c+1/2}$ and $R<Sq<2S^{Q+1}$, we get
\begin{align}
  \frac{\|\mathbf z\mathbf E\|_p}{r_0}
  &\ge\frac{L^{1/p}\sqrt{T+2}}{L^{1/p}R^{1/p}}
    =\frac{\sqrt{T+2}}{R^{1/p}}\notag\\
  &>2^{1-1/p}S^{\,c+1/2-(Q+1)/p}
    =2^{1-1/p}S^{c_H},
  \label{eq:heavy-fiber-ratio}
\end{align}
where
\[
  c_H=c+\frac12-\frac{Q+1}{p}=(2c+4)\vartheta-\frac32.
\]

\paragraph{Obtaining the Claimed Gap.}
Since $r<2r_0$, dividing the bounds in
\eqref{eq:zero-residue-ratio}, \eqref{eq:large-support-ratio}, and
\eqref{eq:heavy-fiber-ratio} by two gives the following strict lower
bounds on $\|\mathbf z\mathbf E\|_p/r$ in Cases 1, 2, and 3, respectively,
\begin{equation}\label{eq:post-rounding-three}
  \frac12S^{c_0},\qquad
  8^{-\vartheta}M^{(c-2)/(2D)},\qquad
  2^{-1/p}S^{c_H}.
\end{equation}
We claim that each of these quantities exceeds $M^{\chi/D}$. By
\eqref{eq:parameter-family},
\begin{equation}\label{eq:soundness-margins}
  c_0-\chi\ge c+\frac{11}{4},\qquad
  \frac{c-2}{2}-\chi\ge\frac14,\qquad
  c_H-\chi\ge\frac14,\qquad
  0<\frac\chi D<\frac18.
\end{equation}
To compare powers of $S$ with powers of $M$, recall from
\eqref{eq:R-over-g} that
\begin{equation}\label{eq:rank-gap-bounds}
  M<40S^D,\qquad M>\frac{S^{2Q}}2>2^{6D}.
\end{equation}
The last inequality uses $S\ge102>2^{20/3}$, $Q\ge11$, and $D=2Q+1$, so
\[
  \frac{S^{2Q}}2>2^{40Q/3-1}>2^{12Q+6}=2^{6D}.
\]

For Cases 1 and 3 we use the upper bound
$M^{\chi/D}<40^{1/8}S^\chi<2S^\chi$.
In Case 1 the first inequality of \eqref{eq:soundness-margins} gives
$S^{c_0-\chi}/2>2$. In Case 3 we have
$2^{-1/p}S^{c_H-\chi}>2$ since $p>2$ and $S>64$.
For Case 2, divide its bound in \eqref{eq:post-rounding-three} by
$M^{\chi/D}$. The ratio satisfies
\[
  8^{-\vartheta}M^{((c-2)/2-\chi)/D}
  \ge 8^{-\vartheta}M^{1/(4D)}>1,
\]
since $M^{1/(4D)}>2^{3/2}>8^\vartheta$. This proves in all three cases that
\[
  \lambda_1^{(p)}(\cL(\mathbf B))>M^{\eps_c(p)}r.
\]

\subsection{Encoding the Output in Polynomial Time}\label{sec:complexity}

For an explicit $3$CNF input of bit length $s$, we have $n,m\le s$, so
$100+s\le S\le100+3s$. Fixing $p$ and $c$, all of $q,T,\kappa,M,L$ are
polynomial in $S$ and have $O_{p,c}(\log S)$ bits. We find $q$ by testing
the integers in \eqref{eq:q-range} for primality~\cite{aks}.
There are polynomially many table coordinates and linear constraints.
Gaussian elimination over $\F_q$ computes parity checks for the
Reed--Solomon evaluation spaces and a generator of $\cC_\varphi$ in
systematic form. Equation~\eqref{eq:construction-a-basis} then gives
$\mathbf B_A$. The algebraic nodes, factors, and modules enter only in
the proof of soundness.

The Sylvester matrix is explicit. Each entry of $\mathbf B_A\mathbf E$ is
a sum of at most $q$ signed entries of $\mathbf B_A$ and hence has
magnitude at most $q^2$. Thus the dense basis has polynomial bit length.
Binary search in \eqref{eq:integer-radius} computes the integer radius in
polynomial time, using $O_{p,c}(\log S)$-bit integers.

On malformed inputs and formulas containing an empty clause, output
$\mathbf B=(2)$ and $r=1$. If no clauses remain, output $\mathbf B=(1)$
and $r=1$. These rank-one lattices have shortest-vector lengths two in the NO
case and one in the YES case.
This proves Theorem~\ref{thm:rational-core}.

\section{Proof of the Main Theorem}\label{sec:parameter-range}

\paragraph{Choosing the Exponent.}\label{sec:optimization}
Fix a rational $p>2$ and recall that $\vartheta=1/2-1/p$. The exponent in
Theorem~\ref{thm:rational-core} is the smaller of
\[
  \frac{(2c+4)\vartheta-7/4}{4c+7}
  \qquad\text{and}\qquad
  \frac{(2c-5)/4}{4c+7}.
\]
The first bound comes from a single dense fiber, and the second from the
total support of the residue. Taking $c\to\infty$, we obtain
\begin{equation}\label{eq:F-definition}
  \lim_{c\to\infty}\eps_c(p)
  =\min\left\{\frac\vartheta2,\frac18\right\}
  =\min\left\{\frac{p-2}{4p},\frac18\right\}.
\end{equation}
Thus for any positive $\varepsilon$ below this minimum, a sufficiently
large constant $c$ gives the claimed hardness factor $M^\varepsilon$.

\paragraph{Passing to Real Norms.}\label{sec:real-norms}
Fix a finite real $p>2$ and an exponent $\varepsilon$ below the limit
in~\eqref{eq:F-definition}. Choose a rational $2<p'<p$ so close to $p$ that
the reduction for $p'$ gives some exponent $\varepsilon'>\varepsilon+2\omega$,
where $\omega=1/p'-1/p$. Use its basis and integer radius $r$.
By Lemma~\ref{lem:norm-comparison},
$\|\mathbf y\|_p\le\|\mathbf y\|_{p'}\le N^\omega\|\mathbf y\|_p$.
Completeness follows from the first inequality. Since $N<2M\le M^2$ on
nontrivial inputs, the second gives soundness with exponent
$\varepsilon'-2\omega>\varepsilon$. On the exceptional outputs both norms are absolute value.

\paragraph{The $\ell_\infty$ Norm.}
The satisfying vector has entries in $\{0,\pm1\}$, so we output radius one.
Fix $c\ge4$. For all sufficiently large rational $p'$, the smaller term in
$\eps_c(p')$ is $\eps_c(\infty):=(2c-5)/(4(4c+7))$. The lattice is the
same for every such $p'$, and its finite-norm radius is at least two since
$RL>1$ on nontrivial inputs. Thus on a nontrivial NO instance every nonzero lattice
vector $\mathbf y$ satisfies $\|\mathbf y\|_{p'}>2M^{\eps_c(\infty)}$.
Taking $p'\to\infty$ gives $\|\mathbf y\|_\infty\ge2M^{\eps_c(\infty)}$.
The rank-one outputs already have the required gap. Since
$\eps_c(\infty)\to1/8$ as $c\to\infty$, this proves
Theorem~\ref{thm:main-introduction}.

\bibliographystyle{alpha}
\bibliography{refs}

\end{document}